\documentclass[times, 11pt]{shrinkarticle}
\usepackage{latex8}
\usepackage{times,paralist}
\usepackage{gastex}

\usepackage{amsmath,amssymb,amsfonts,amsthm}

\usepackage{multicol}

\newtheorem{definition}{Definition}
\newtheorem{theorem}{Theorem}
\newtheorem{proposition}{Proposition}
\newtheorem{corollary}{Corollary}
\newtheorem{problem}{Problem}
\newtheorem{lemma}{Lemma}
\newtheorem{example}{Example}
\newtheorem{remark}{Remark}
\newtheorem{algorithm}{Algorithm}

\usepackage[all]{xy}

\newcommand{\Inf}{\mathrm{Inf}}
\newcommand{\Reach}{{\text{\textrm{Reach}}}}
\newcommand{\Parity}{{\mathrm{Parity}}}
\newcommand{\CRec}{{\mathrm{CRec}}}

\newcommand{\Supp}{\textrm{Supp}}
\newcommand{\Buchi}{\textrm{B\"uchi}}
\newcommand{\Safe}{\textrm{Safe}}
\newcommand{\coBuchi}{\textrm{coB\"uchi}}
\newcommand{\set}[1]{\{\: #1 \:\}}

\makeatletter

\begingroup \catcode `|=0 \catcode `[= 1
\catcode`]=2 \catcode `\{=12 \catcode `\}=12
\catcode`\\=12 |gdef|@xcomment#1\end{comment}[|end[comment]]
|endgroup

\def\@comment{\let\do\@makeother \dospecials\catcode`\^^M=10\def\par{}}

\def\begincomment{\@comment\@xcomment}

\makeatother

\newenvironment{comment}{\begincomment}{}

\title{Decidable Problems for Probabilistic Automata on Infinite Words\\}

\author{Krishnendu Chatterjee (IST Austria) \and Mathieu Tracol (IST Austria) 
}

\begin{document}

\maketitle
\begin{center}

 
\end{center}
 

\begin{abstract}

We consider probabilistic automata on infinite words with 
acceptance defined by parity conditions.
We consider three qualitative decision problems: 
(i)~the \emph{positive} decision problem asks whether there is a word that 
is accepted with positive probability; 
(ii)~the \emph{almost} decision problem asks whether there is a word that 
is accepted with probability~1; and 
(iii)~the \emph{limit} decision problem asks whether for every $\epsilon>0$ 
there is a word that is accepted with probability at least $1-\epsilon$.
We unify and generalize several decidability results for probabilistic automata 
over infinite words, and identify a \emph{robust} (closed under 
union and intersection) subclass of probabilistic automata 
for which all the qualitative decision problems are decidable for parity 
conditions. 
We also show that if the input words are restricted to lasso shape (regular) words, 
then the positive and almost problems are decidable for all probabilistic 
automata with parity conditions.
For most decidable problems we show an optimal PSPACE-complete complexity
bound.
\end{abstract}


\newpage

\setcounter{page}{1}
 
\section{Introduction}

\noindent{\bf Probabilistic automata.}
The class of probabilistic automata for finite words was introduced in the 
seminal work of Rabin~\cite{RabinProb63} as an extension of classical 
finite automata.
Probabilistic automata on finite words  have been extensively studied 
(see the book~\cite{paz1971introduction} on probabilistic automata and the survey 
of~\cite{Bukharaev}).
Probabilistic automata on infinite words have been studied recently 
in the context of verification and analysis of reactive 
systems~\cite{baier2005recognizing,baier2008decision,chadha2009power,CSV09}.
We consider probabilistic automata on infinite words with 
acceptance defined by safety, reachability, B\"uchi, 
coB\"uchi, and parity conditions, as they can express all commonly 
used specifications (like safety, liveness, fairness) of verification.

\smallskip\noindent{\bf Qualitative decision problems.}
We consider three \emph{qualitative} decision 
problems for probabilistic automata on infinite words~\cite{baier2008decision,gimbert2010probabilistic}:
given a probabilistic automaton with an acceptance condition, 
(i)~the \emph{positive} decision problem asks whether there is a word that 
is accepted with positive probability (probability $>0$); 
(ii)~the \emph{almost} decision problem asks whether there is a word that 
is accepted almost-surely (with probability~1); and 
(iii)~the \emph{limit} decision problem asks whether for every $\epsilon>0$ 
there is a word that is accepted with probability at least $1-\epsilon$.
The qualitative decision problems for probabilistic automata are the 
generalization of the emptiness and universality problems for deterministic 
automata.

\smallskip\noindent{\bf Decidability and undecidability results.} 
The decision problems for probabilistic automata on finite words have been 
extensively studied~\cite{paz1971introduction,Bukharaev}, and the main results establish the
undecidability of the quantitative version of the decision problems 
(where the thresholds are a rational $0<\lambda<1$, rather than $0$ and $1$).
The undecidability results for the qualitative decision problems for 
probabilistic automata on infinite words are quite recent. 
The results of~\cite{baier2008decision} show that the positive (resp. almost) decision 
problem is undecidable for probabilistic automata with B\"uchi (resp. 
coB\"uchi) acceptance condition, and as a corollary both the positive and 
almost decision problems are undecidable for  parity acceptance conditions 
(as both B\"uchi and coB\"uchi conditions are special cases of parity 
conditions). 
The results of~\cite{baier2008decision} also show that the positive (resp. almost) decision 
problem is decidable for probabilistic automata with coB\"uchi (resp. 
B\"uchi) acceptance condition, and these results have been extended to 
the more general case of stochastic games with imperfect information 
in~\cite{bertrand2009qualitative} and~\cite{gripon2009qualitative}. 
The positive and almost problems are decidable for safety and reachability 
conditions, and also for probabilistic automata over finite words.
For all the decidable almost and positive problems for probabilistic automata 
PSPACE-complete bounds were established in~\cite{CSV09,chadha2009power}.
It was shown in~\cite{gimbert2010probabilistic} that the limit decision problem is undecidable 
even for probabilistic finite automata, and the proof can be easily adapted 
to show that the limit decision problem is undecidable for reachability,
B\"uchi, coB\"uchi and parity conditions (see~\cite{chatterjee2010} for details).

\smallskip\noindent{\bf Decidable subclasses.} The root cause of the 
undecidability results is that for arbitrary probabilistic automata and 
arbitrary input words the resulting probabilistic process is 
complicated. 
As a consequence several researchers have focused on identifying subclasses of 
probabilistic automata where the qualitative decision problems are decidable. 
The work of~\cite{chadha2009power} presents a subclass of probabilistic automata, namely 
hierarchical probabilistic automata (HPA), and show that the positive and 
almost problems are decidable for B\"uchi and coB\"uchi conditions on HPAs. 
The work of~\cite{gimbert2010probabilistic} presents a subclass of probabilistic automata, namely 
$\#$-acyclic automata, and show that the limit reachability problem is decidable for 
this class of automata over finite words. 
The two subclasses HPA and $\#$-acyclic automata are incomparable in 
expressive power.

\smallskip\noindent{\bf Our contributions.} In this work we unify and 
generalize several decidability results for probabilistic automata over 
infinite words, and identify a robust subclass of probabilistic automata 
for which all the qualitative decision problems are decidable for parity 
acceptance conditions. 
For the first time, we study the problem of restricting the structure of 
input words, as compared to the probabilistic automata, and show that 
if the input words are restricted to \emph{lasso shape} words, then the positive
and almost problems are decidable for all probabilistic automata with 
parity acceptance conditions. 
The details of our contributions are as follows.

\begin{compactenum}

\item We first present a very general result that would be the basic 
foundation of the decidability results.
We introduce a notion of \emph{simple} probabilistic process: 
the non-homogeneous Markov chain induced on the state space of a 
probabilistic automaton by an infinite word is simple if the tail 
$\sigma$-field of the process has a particular structure. 
The structure of the tail $\sigma$-field is derived from Blackwell-Freedman-Cohn-Sonin
\emph{decomposition-separation theorem}~\cite{blackwell1964tail,cohn1989products,sonin1996asymptotic} 
on finite non-homogeneous 
Markov chains which generalizes the classical results on homogeneous 
Markov chains.

\item We then show that if we restrict the input words of a probabilistic 
automaton to those which induce simple processes, then the positive and almost 
decision problems are decidable for parity conditions. 
We establish that these problems are PSPACE-complete.

\item We study for the first time the effect of restricting the 
structure of input words for probabilistic automata, rather than 
restricting the structure of probabilistic automata. 
We show that for all ultimately periodic (regular or lasso shape) words and for all 
probabilistic automata, the probabilistic  process induced is a simple one. 
Hence as a corollary of our first result, we obtain that if we restrict to 
lasso shape words, then the positive and almost decision problems are 
decidable (PSPACE-complete) for all probabilistic automata with parity conditions.
However, the limit decision problem for the reachability condition 
is still undecidable for lasso shape 
words, as well as for the
B\"uchi and coB\"uchi conditions.

\item We then introduce the class of \emph{simple probabilistic automata} 
(for short simple automata): a probabilistic 
automaton is simple if every input infinite words induce simple processes on its 
state space. 
This semantic definition of simple automata uses the decomposition-separation theorem.
We present a \emph{structural (or syntactic)} subclass of the class of 
simple automata, called \emph{structurally simple automata}, which relies on the structure 
of the \emph{support graph} of the automata (the support graph is obtained via
subset constructions of the automata). 
We show that the class of structurally simple automata generalizes both the models of HPA and $\#$-acyclic 
automata.
Since HPA generalizes deterministic automata, it follows that structurally simple 
automata with parity conditions strictly generalizes $\omega$-regular languages.  
We show that for structurally simple automata with parity conditions, the positive and almost 
problems are PSPACE-complete, and the limit problem can be decided in EXPSPACE.
Thus our results both unify and generalize two different results for 
decidability of subclasses of probabilistic automata.
Moreover, we show that structurally simple automata are \emph{robust}, i.e., closed under 
union and intersection.
Thus we are able to identify a robust subclass of probabilistic automata 
for which all the qualitative decision problems are decidable for parity 
conditions.
From our structural characterization it also follows that given a probabilistic automaton, it can be
decided in EXPSPACE whether the automaton is structurally simple.

\end{compactenum}
In this paper we use deep results from probability theory to establish 
general results about the decidability of problems on probabilistic automata. 
We present a sufficient structural condition to ensure semantic notions (of an 
induced probabilistic process being simple) coming from probability theory in the 
context of probabilistic automata. The proofs omitted due to lack of space 
are given in appendix.

\section{Preliminaries}


\noindent{\bf Distributions.} Given a finite set $Q$, we denote by $\Delta(Q)$ the set of probability distributions on $Q$. 
Given $\alpha\in\Delta(Q)$, we denote by $\mathrm{Supp}(\alpha)$ the support of $\alpha$, i.e. $\mathrm{Supp}(\alpha)=\lbrace q\in Q\ |\ \alpha(q)>0\rbrace$. 

\smallskip\noindent{\bf Words and prefixes.}
Let $\Sigma$ be a finite \emph{alphabet} of letters. 
A \emph{word} $w$ is a finite or infinite sequence of letters from 
$\Sigma$, i.e., $w \in\Sigma^*$ or $w \in \Sigma^\omega$.
Given a word $w=a_1,a_2...\in\Sigma^\omega$ and $i\in\mathbb{N}$, we define $w(i)=a_i$, 
and we denote by $w[1..i]=a_1,...,a_i$ the prefix of length $i$ of $w$. 
Given $j\geq i$, we denote by $w[i..j]=a_i,...,a_j$ the subword of $w$ 
from index $i$ to $j$.
An infinite word $w\in\Sigma^\omega$ is a \emph{lasso shape word} if there exist 
two finite words $\rho_1$ and $\rho_2$ in $\Sigma^*$ such that $w=\rho_1\cdot\rho_2^\omega$.

\begin{definition}[Finite Probabilistic Table (see \cite{paz1971introduction})]
A \emph{Finite Probabilistic Table} (FPT) is a tuple $\mathcal{T}=(Q,\Sigma,\lbrace M_a\rbrace_{a\in\Sigma},\alpha)$ 
where $Q$ is a finite set of states, $\Sigma$ is a finite alphabet, $\alpha$ is an initial distribution on $Q$, 
and the $M_a$, for $a\in\Sigma$, are Markov matrices of size $|Q|$, 
i.e., for all $q, q' \in Q$ we have $M_a(q,q') \geq 0$
and for all $q \in Q$ we have $\sum_{q' \in Q} M_a(q,q') =1$.
\end{definition}

\noindent{\bf Distribution generated by words.}
For a letter $a \in \Sigma$, let $\delta(q,a)(q')=M_a(q,q')$ 
denote the transition probability from $q$ to $q'$ given the 
input letter $a$.
Given $\beta\in\Delta(Q)$, $q\in Q$ and $\rho\in\Sigma^*$, let $\delta(\beta,\rho)(q)$ be the probability, 
starting from a state sampled accordingly to $\beta$ and reading the input word $\rho$, to go to state $q$. 
Formally, given $\rho=a_1,...,a_n\in\Sigma^*$, let $M_\rho=M_{a_1}\cdot M_{a_2}\cdot...\cdot M_{a_n}$. 
Then $\delta(\beta,\rho)(q)=\sum_{q'\in Q}\beta(q')\cdot M_{\rho}(q',q)$.
We often write $\delta(\beta,\rho)$ instead of $\mathrm{Supp}(\delta(\beta,\rho))$, 
for simplicity: $\delta(\beta,\rho)$ is the set of states reachable with positive probability when starting 
from distribution $\beta$ and reading $\rho$. As well, given $H\subseteq Q$, we write $\delta(H,\rho)$ 
for the the set of states reachable with positive probability when starting from a state in 
$H$ sampled uniformly at random, and reading $\rho$.

\smallskip\noindent{\bf Homogeneous and non-homogeneous Markov chains.} A Markov chain is a sequence of random variables $X_0, X_1, X_2,...$, taking values in a (finite) set $Q$, with the Markov property: $\mathbb{P}(X_{n+1}=x|X_1=x_1, X_2=x_2, \ldots, X_n=x_n) = \mathbb{P}(X_{n+1}=x|X_n=x_n)$. Given $n\in\mathbb{N}$, the matrix $M_n$ of size $|Q|$ such that for all $q,q'\in Q$ we have $M_n(q,q')=\mathbb{P}(X_{n+1}=q'|X_n=q)$ is the \emph{transition matrix at time $n$} of the chain. The Markov chain is \emph{homogeneous} if $M_n$ does not depend on $n$. In the general case, we call the chain \emph{non-homogeneous}.

\smallskip\noindent{\bf Induced Markov chains.}
Given a FPT with state space $Q$, given $G\subseteq Q$ and $\rho=a_0,...,a_{m-1}\in\Sigma^*$ such that $\delta(G,\rho)\subseteq G$, we define the Markov chain $\lbrace X_n\rbrace_{n\in\mathbb{N}}$ induced by $(G,\rho)$ as follows: the initial distribution, i.e. the distribution of $X_0$, is uniform on $G$; given $i\in\mathbb{N}$, $X_{i+1}$ is distributed according to $\delta(X_i,a_{i\ \mathrm{mod}\ m})(-)$. Intuitively, $\lbrace X_n\rbrace_{n\in\mathbb{N}}$ is the Markov chain induced on the FPT when reading the word $\rho^\omega$.

\smallskip\noindent{\bf Probability space and $\sigma$-field.}
A word $w\in\Sigma^\omega$ induces a probability space $(\Omega,\mathcal{F},\mathbb{P}^w)$: $\Omega=Q^\omega$ 
is the set of \emph{runs}, $\mathcal{F}$ is the $\sigma$-field generated by cones of the type 
$C_\rho=\lbrace r\in Q^\omega\ |\ r[1..|\rho|]=\rho\rbrace$ where $\rho\in Q^*$, and $\mathbb{P}^{w}$ is the associated probability 
distribution on $\Omega$. See \cite{vardi1985avp} for the standard results on this topic. 
We write $\lbrace X_n^w\rbrace_{n\in\mathbb{N}}$ for the \emph{non-homogeneous} Markov chain induced 
on $Q$ by $w$, and given $n\in\mathbb{N}$ let $\mu_n^w$ be the distribution of $X_n^w$ on $Q$:
\[\mathrm{Given}\ q\in Q,\ \  \mu_n^w(q)=\mathbb{P}^{w}[\lbrace r\in\Omega\ |\ r(n)=q\rbrace]\]
The $\sigma$-field $\mathcal{F}$ is also the smallest $\sigma$-field on $\Omega$ with respect to which 
all the $X_n^w,\ n\in\mathbb{N}$, are measurable. For all $n\in\mathbb{N}$, 
let $\mathcal{F}_n=\mathcal{B}(X_n^w,X_{n+1}^w,...)$ be the smallest $\sigma$-field on $\Omega$ with 
respect to which all the $X_i^w,\ i\geq n$, are measurable. 
We define $\mathcal{F}_\infty=\bigcap_{n\in\mathbb{N}}\mathcal{F}_n$, called the \emph{tail $\sigma$-field of 
$\lbrace X_n^w\rbrace$}. 
Intuitively, an event $\Gamma$ is in $\mathcal{F}_\infty$ if changing a finite number of states of a 
run $r$ does not affect the occurrence of the run $r$ in $\Gamma$.

\smallskip\noindent{\bf Atomic events.} 
Given a probability space $(\Omega,\mathcal{F},\mathbb{P})$ 
and $\Gamma\in\mathcal{F}$, we say that $\Gamma$ is 
\emph{$\mathcal{F}$-atomic} if $\mathbb{P}(\Gamma)>0$, and for all $\Gamma'\in\mathcal{F}$ such that 
$\Gamma'\subseteq\Gamma$ we have either $\mathbb{P}(\Gamma')=0$ or $\mathbb{P}(\Gamma')=\mathbb{P}(\Gamma)$. 
In this paper we will use atomic events in relation to the tail $\sigma$-field of Markov chains.

\smallskip\noindent{\bf Acceptance conditions.}
Given a FPT, let $F \subseteq Q$ be a 
set of accepting (or target) states. 
Given a run $r$, we denote by $\Inf(r)$ the set of 
states that appear infinitely often in $r$.
We consider the following acceptance conditions.
\begin{compactenum}

\item \emph{Safety condition.} 
The safety condition $\Safe(F)$ defines the set of paths 
that only visit states in $F$; i.e., 
$\Safe(F) =\set{(q_0,q_1,\ldots) \mid \forall i \geq 0. \ q_i \in F}$.

\item \emph{Reachability condition.}
The reachability condition $\Reach(F)$ defines the set of paths 
that visit states in $F$ at least once; i.e., 
$\Reach(F) =\set{(q_0,q_1,\ldots) \mid \exists i \geq 0. \ q_i \in F}$.

\item \emph{B\"uchi condition.} 
The B\"uchi condition $\Buchi(F)$ defines the set of paths 
that  visit states in $F$ infinitely often; i.e., 
$\Buchi(F) =\set{r \mid \Inf(r) \cap F \neq \emptyset}$.

\item \emph{coB\"uchi condition.}
The coB\"uchi condition $\coBuchi(F)$ defines the set of paths 
that  visit states outside $F$ finitely often; i.e., 
$\coBuchi(F) =\set{r \mid \Inf(r) \subseteq F}$.

\item \emph{Parity condition.} 
The parity condition consists of a priority function 
$p: Q \to \mathbb{N}$ and defines the set of paths 
such that the minimum priority visited infinitely often 
is even, i.e.,
$\Parity(p) =\set{r \mid \min(p(\Inf(r))) \text{ is even}}$.
B\"uchi and coB\"uchi conditions are special cases of parity conditions 
with two priorities (priority set $\set{0,1}$ for B\"uchi and $\set{1,2}$ 
for coB\"uchi).

\end{compactenum}

\smallskip\noindent{\bf Probabilistic automata.} 
A \emph{Probabilistic Automaton} (PA) is a tuple $\mathcal{A}=(\mathcal{T},\Phi)$ 
where $\mathcal{T}$ is a FPT and $\Phi$ is an acceptance condition.

\smallskip\noindent{\bf Decision problems.} Let $\mathcal{A}$ be a PA with 
acceptance condition $\Phi:\Omega\rightarrow \set{0,1}$.
We consider the following decision problems.
\begin{compactenum}
 \item \emph{Almost problem:} Whether there exists $w\in\Sigma^\omega$ such that $\mathbb{P}^w_\mathcal{A}(\Phi)=1$?
 \item \emph{Positive problem:} Whether there exists $w\in\Sigma^\omega$ such that $\mathbb{P}^w_\mathcal{A}(\Phi)>0$?
 \item \emph{Limit problem:} Whether for all $\epsilon>0$, there exists $w\in\Sigma^\omega$ such that $\mathbb{P}^w_\mathcal{A}(\Phi)>1-\epsilon$?
\end{compactenum}


\noindent Proposition \ref{p-summ known} summarizes the known results from~\cite{baier2008decision,chatterjee2010,gimbert2010probabilistic,CSV09,chadha2009power}.

\begin{proposition}
\label{p-summ known} 
Given a PA with an acceptance condition $\Phi$, the following assertions hold:
\begin{compactenum}
\item  The almost problem is decidable (PSPACE-complete) for $\Phi=$ safety, reachability, B\"uchi, and undecidable for $\Phi=$ co-B\"uchi and parity.
 \item The positive problem is decidable (PSPACE-complete) for $\Phi=$ safety, reachability, co-B\"uchi, and undecidable for $\Phi=$ B\"uchi and parity.
 \item The limit problem is decidable (PSPACE-complete) for $\Phi=$ safety, and undecidable for $\Phi=$ reachability, B\"uchi, co-B\"uchi, and parity.
\end{compactenum}
\end{proposition}

\section{Simple Processes}\label{s-simple process}

In this section we first recall the decomposition-separation theorem, 
then use it to decompose the tail $\sigma$-field of stochastic processes into atomic events.
We then introduce the notion of simple processes, which are stochastic 
processes where the atomic events obtained using the decomposition-separation
theorem are \emph{non-communicating}.

\subsection{The Decomposition Separation Theorem and tail $\sigma$-fields}

The structure of the tail $\sigma$-field of a general non-homogeneous Markov chain 
has been deeply studied by mathematicians.
Blackwell and Freedman, in \cite{blackwell1964tail}, presented a generalization of the classical \emph{decomposition theorem} 
for homogeneous Markov chains, in the context of non-homogeneous Markov chains with finite state spaces. 
The work of Blackwell and Freedman has been deepened by Cohn \cite{cohn1989products} and Sonin \cite{sonin1996asymptotic}, 
who gave a more complete picture. 
We present the results of \cite{blackwell1964tail,cohn1989products,sonin1996asymptotic} in the framework of 
\emph{jet decompositions} presented in \cite{sonin1996asymptotic}.

\smallskip\noindent{\bf Jets and partition into jets.}
A \emph{jet} is a sequence $J=\{J_i\}_{i \in \mathbb{N}}$, where each $J_i \subseteq Q$.
A tuple of jets $(J^0,J^1,...,J^c)$ is called a \emph{partition of $Q^\omega$ into jets} if for 
every $n\in\mathbb{N}$, we have that $J^0_n,J^1_n,...,J^c_n$ is a partition of $Q$. 
The \emph{Decomposition-Separation Theorem}, in short \emph{DS-Theorem}, proved by Cohn~\cite{cohn1989products} and Sonin~\cite{sonin1996asymptotic} 
using results of \cite{blackwell1964tail}, 
is given in Theorem~\ref{t-decomp}.
We first define the notion of \emph{mixing} property of jets. 

\smallskip\noindent{\bf Mixing property of jets.} 
Given a FPT $\mathcal{A}$, a jet $J=\{J_i\}_{i\in \mathbb{N}}$ is \emph{mixing} for a word $w$
if: given $X^w_n,n\geq0$ the process induced on $Q$ by $w$, given $q,q'\in Q$, and a sequence of states $\lbrace q_i\rbrace_{i\in\mathbb{N}}$ such that for all 
$i \geq 0$ we have $q_i\in J_i$, given $m\in\mathbb{N}$, if $\mathrm{lim}_n\mathbb{P}^w[X^w_n=q_n\ |\ X^w_m=q]>0$ and $\mathrm{lim}_n\mathbb{P}^w[X^w_n=q_n\ |\ X^w_m=q']>0$,
then we have:
\[\mathrm{lim}_{n\rightarrow\infty}\dfrac{\mathbb{P}^w[X^w_n=q_n\ |\ X^w_n\in J^k_n\wedge X^w_m=q]}{\mathbb{P}^w[X^w_n=q_n\ |\ X^w_n\in J^k_n\wedge X^w_m=q']}=1\]

Intuitively, a jet is mixing if the probability distribution of a state of the process, conditioned to the fact that this state belongs to the jet, is ultimately independent of the initial state. This extends the notion of mixing process on homogeneous ergodic Markov chains, on which the distribution of a state of the process after a number of steps is close to the stationary distribution, irrespective of the initial state.

\begin{theorem}[The Decomposition-Separation (DS) Theorem \cite{blackwell1964tail,cohn1989products, sonin1996asymptotic}]\label{t-decomp}
Given a FPT $\mathcal{A}=(Q,\Sigma,\lbrace M_a\rbrace_{a\in\Sigma},\alpha)$, 
for all $w\in\Sigma^\omega$ there exists $c\in\set{1,2,\ldots,|Q|}$ and a partition $(J^0,J^1,...,J^c)$ of $Q^\omega$ into jets 
such that:
\begin{compactenum}
 \item With probability one, after a finite number of steps, a run $r\in\Omega$ enters into one of the jets $J^k$, $k\in\set{1,2,\ldots,c}$ and stays there forever.
 \item For all $k \in \set{1,2,\ldots,c}$ the jet $J^k$ is mixing. 
\end{compactenum}
\end{theorem}

Theorem~\ref{t-decomp} holds even if $\Sigma$ is infinite: it is valid for any non-homogeneous Markov 
chain on a finite state space. 
In this paper we will focus on finite alphabets only.

\smallskip\noindent{\bf Remark.}
We note that for all $i\in \set{1,2,\ldots,c}$,
 either $\mu_n^w(J^i_n)\rightarrow_{n\rightarrow\infty}0$ or there exists $\lambda_i>0$ such that for $n$ 
large enough $\mu_n^w(J^i_n)>\lambda_i$. Indeed, if $\mu_n^w(J^i_n)\not\rightarrow_{n\rightarrow\infty}0$ but 
there exists a subsequence of $\lbrace \mu_n^w(J^i_n)\rbrace_{n\in\mathbb{N}}$ which goes to zero, 
then a non zero probability of runs enter $J^i_n$ and leave it afterward infinitely often, which contradicts
the first point of Theorem~\ref{t-decomp}.
Thus, we can always assume that there exists $\lambda>0$ such that for all $i\in\set{1,2,\ldots,c}$, 
for $n$ large enough, we have $\mu^w_n(J^i_n)>\lambda$. If this is not the case, we just merge 
the jets $J^i$ such that $\mu_n^w(J^i_n)\rightarrow_{n\rightarrow\infty}0$ with $J^0$, which does not 
invalidate the properties of the jet decomposition stated by Theorem \ref{t-decomp}.

For the following of the section, we fix $w\in\Sigma^\omega$ and a partition $J^0,J^1,...,J^c$ of $Q^\omega$ as in the DS Theorem. Given $i\in\set{1,2,\ldots,c}$ and $n\in\mathbb{N}$, let:
\[\tau^i_n=\lbrace r\in\Omega\ |\ r(i)\in J^i_n\rbrace,\ \ \mathrm{and}\ \ \tau_\infty^i=\cup_{N\in\mathbb{N}}\cap_{n\geq N}\tau^i_n\]
We now present a result directly from our formulation of the DS Theorem 
(the result can also be proved using more general results of \cite{cohn1989products}). 

\begin{proposition}\label{p-tau atomics}
For all $i\in\set{1,2,\ldots,c}$, the following assertions hold:
(1)~$\tau_\infty^i\in\mathcal{F}_\infty$, i.e., $\tau^\infty_i$ is a tail $\sigma$-field 
event;
(2)~$\tau_\infty^i$ is $\mathcal{F}_\infty$-atomic; i.e.,
$\tau^\infty_i$ is an atomic tail event; and
(3)~$\mathbb{P}^{w}(\bigcup_{i=1}^c\tau^i_\infty)=1$.
\end{proposition}

The fact that the $\tau^i_\infty$ are atomic sets of $\mathcal{F}_\infty$ means that all the runs which belong to 
the same $\tau^i_\infty$ will satisfy the same \emph{tail} properties. 
Intuitively, a tail does not depend on finite prefixes.
Several important classes of properties are tail properties, as presented in \cite{de2009statistic}: in particular 
any parity condition is a tail property.


\subsection{Simple processes characterization with jets}

\begin{definition}
Let $\lbrace X^w_n\rbrace_{n\in\mathbb{N}}$ be a process induced on $Q$ by a word $w\in\Sigma^\omega$, and let $\mu^w_n$ be its probability distribution on $Q$ at time $n$. We say that $\lbrace \mu^w_n\rbrace_{n\in\mathbb{N}}$ is \emph{simple} if there exist $\lambda>0$ and two sequences $\lbrace A_n\rbrace_{n\in\mathbb{N}}$ and $\lbrace B_n\rbrace_{n\in\mathbb{N}}$ of subsets of $Q$ such that:
\begin{compactitem}
 \item $\forall n\in\mathbb{N}$, $A_n,B_n$ is a partition of $Q$
 \item $\forall n\in\mathbb{N},\ \forall q\in A_n,\ \ \mu^w_n(q)>\lambda$
 \item $\mu^w_n(B_n)\rightarrow_{n\rightarrow\infty}0$
\end{compactitem}
\end{definition}

The second point of the following proposition shows that the tail $\sigma$-field of a simple process can be decomposed as a set of ``non-communicating'' jets. Intuitively, a jet is non-communicating if there exists a bound $N\in\mathbb{N}$ such that after time $N$, if a run belongs to the jet, it will stay in it for ever with probability one. The following proposition is a reformulation of the notion of simple process in the framework of jets decomposition.

\begin{proposition}\label{p-separation simple jets}
Let $w\in\Sigma^\omega$, and suppose that the process $\lbrace \mu^w_n\rbrace_{n\in\mathbb{N}}$ induced on $Q$ is simple. Then there exists a decomposition of $Q^\omega$ into jets, $J^0,J^1,...,J^c$, and $N\in\mathbb{N}$, which satisfy the following properties:
\begin{compactenum}
 \item For all $n\geq N$, all $i\in\set{1,2,\ldots,c}$ and all $q\in J^i_n$, we have $\mu_n(q)>\lambda$.
 \item For all $i\in\set{1,2,\ldots,c}$ and all $n_2>n_1\geq N$ we have $\delta(J^i_{n_1},w^{n_2}_{n_1+1})\subseteq J^i_{n_2}$.
 \item $\mu^w_n(J^0_n)\rightarrow_{n\rightarrow\infty}0$.
 \item Each jet $J^i,\ i\in\set{1,2,\ldots,c}$ is mixing.
\end{compactenum}
\end{proposition}

\section{Decidable Problems for Simple Processes and Lasso shape Words}\label{s-decidable pb}
In this section we will present decidable algorithms (with optimal complexity) 
for the decision problems with the restriction of simple processes, and for lasso shape words.

\subsection{Decidable problems for simple processes} 
We first define the \emph{simple decision problems} that impose the 
simple process restriction.
Given an acceptance condition $\Phi$, we consider the following problems:
\begin{enumerate}
 \item \emph{Simple almost (resp. positive) problems:} Does there exist $w\in\Sigma^\omega$ such that $\lbrace \mu_n^w\rbrace_{n\in\mathbb{N}}$ is simple 
and $\mathbb{P}^w_\mathcal{A}(\Phi)=1$ (resp. $\mathbb{P}^w_\mathcal{A}(\Phi)>0$)?
 \item \emph{Simple limit problem:} For all $\epsilon>0$, is there $w\in\Sigma^\omega$ such that $\lbrace \mu_n^w\rbrace_{n\in\mathbb{N}}$ is simple and $\mathbb{P}^w_\mathcal{A}(\Phi)>1-\epsilon$?
\end{enumerate}

Proposition~\ref{prop:dec-undec-simple} shows that the decidability and undecidability results of Proposition \ref{p-summ known} concerning the positive, almost, 
and limit safety and reachability problems still hold when we consider their ``simple process'' version. Propositions \ref{p-dec simple almost} and \ref{p-dec simple positive} are more interesting as they 
show that the almost and positive parity problem become decidable when restricted to simple processes. Finally, Proposition \ref{p-undec limit buchi} shows that the ''limit'' decision problems 
remain undecidable even when restricted to simple processes.

\begin{proposition}\label{prop:dec-undec-simple}
The simple almost (resp. positive) safety and reachability problems are PSPACE-complete, as well as the simple limit safety problem. The simple limit reachability problem is undecidable.
\end{proposition}

\begin{proposition}\label{p-dec simple almost}
 The simple almost parity problem is PSPACE-complete
\end{proposition}
\begin{proof} {\em (Sketch).} The proof relies on the following 
equivalent formulation.

\smallskip\noindent{\bf Equivalent formulation.}
In the following, $p:Q\rightarrow\mathbb{N}$ is a parity function on $Q$, and $\Phi=\Parity(p)$. We prove that: \textbf{(1)} There exists $w\in\Sigma^\omega$ such that the induced process is simple and $\mathbb{P}^w_\mathcal{A}(\Phi)=1$ if and only if \textbf{(2)} There exists $G\subseteq Q$ and $\rho_1,\rho_2\in\Sigma^*$ such that $G=\delta(\alpha,\rho_1)$, $\delta(G,\rho_2)\subseteq G$, and the runs on the Markov chain induced by $(G,\rho_2)$ satisfy $\Phi$ with probability one.
We show in the appendix that the properties can be verified in PSPACE and also present a PSPACE lower bound.

We show the equivalence \textbf{(2)}$\Leftrightarrow$\textbf{(1)}. The way \textbf{(2)}$\Rightarrow$\textbf{(1)} is direct, since we will show in Section \ref{s-simple-words} that the process induced by a lasso shape word on any automaton is always simple. We prove that \textbf{(1)}$\Rightarrow$\textbf{(2)}. Let $w=a_1,...,a_i,...$ be such that the induced process is simple and $\mathbb{P}^w_\mathcal{A}(\Phi)=1$. Using Proposition \ref{p-separation simple jets}, let $J^0,J^1,...,J^m$ be the decomposition of $Q^\omega$ into jets and let $N_0\in\mathbb{N},\ \lambda>0$ be such that:
\begin{compactitem}
 \item $\forall n\geq N_0$, $\forall i\in\set{1,2,\ldots,c}$, $\forall q\in J^i_n$: $\mu_n(q)>\lambda$.
 \item $\forall i\in\set{1,2,\ldots,c}$, for all $n_2>n_1\geq N_0$, we have $\delta(J^i_{n_1},w^{n_2}_{n_1+1})\subseteq J^i_{n_2}$.
 \item $\mu^w_n(J^0_n)\rightarrow_{n\rightarrow\infty}0$
 \item Each jet $J^i,\ i\in\set{1,2,\ldots,c}$ is mixing.
\end{compactitem}
Without loss on generality, since $Q$ is finite, taking $N_0$ large enough, we can assume that the vector of sets of states $(J_{N_0}^0,...,J_{N_0}^c)$ appears infinitely often in the sequence $\lbrace(J_{n}^0,...,J_{n}^c)\rbrace_{n\in\mathbb{N}}$. As well, without loss on generality, we can assume that for all $n\geq N_0$ and all $i\in\set{1,2,\ldots,c}$, all the states in $J^i_n$ appear infinitely often among the sets $J^i_m$, for $m\geq N_0$. Let $i\in\set{1,2,\ldots,c}$. Given $q\in Q$, let 
\[\Phi_q=\lbrace r\in\Omega\ |\ q\in\mathrm{Inf}(r)\ \mathrm{and}\ p(q)=\min_{q'\in\mathrm{Inf}(r)}p(q')\rbrace\]
Clearly, for all $q\in Q$, $\Phi_q\in\mathcal{F}_\infty$. Since $Q$ is finite, there exists $q_i\in Q$ such that $\mathbb{P}(\tau^i_\infty\cap \Phi_{q_i})>0$. By Proposition \ref{p-tau atomics}, $\tau_\infty^i$ is atomic, hence $\tau_\infty^i\subseteq \Phi_{q_i}$. Since the runs of the process satisfy the parity condition with probability one, $p(q_i)$ must be even. Moreover, for all $n\geq N_0$ and all $q\in J^i_n$, we must have $p(q)\geq p(q_i)$. Indeed, such a $q$ appears an infinite number of times in the sequence $J^i_n$, by hypothesis, and always with probability at least $\lambda$. 

Since $\tau_\infty^i\subseteq \Phi_{q_i}$, there exists $m_i\in\mathbb{N}$ such that for all $q\in J_i^{N_0}$, there exists $m<m_i$ such that $\delta(q,w[N_0+1..m])(q_i)>0$. We define $m=\max_{i\in\set{1,2,\ldots,c}}\ m_i$, and $m'\geq m$ such that 
\[(J_{N_0}^0,...,J_{N_0}^c)=(J_{N_0+m'}^0,...,J_{N_0+m'}^c)\]
Taking $\rho_1\!\!=\!\!w[0..N_0]$ and $\rho_2\!=\!w[N_0+1..N_0+m']$ completes the proof. Indeed, when starting from the initial distribution, after reading $\rho_1$, we arrive by construction in one of the sets $J^i_{N_0}$, with $i\in\set{0,\ldots,c}$. Starting from this state $q$, if the word $\rho_2$ is taken as input, we go to set $J^i_{N_0+m'}$ with probability one, visit $q_i$ with positive probability, and do not visit any state with probability smaller that $p(q_i)$. This implies that when starting from $q$ and reading $\rho_2^\omega$, we visit $q_i$ with probability one, hence the result.
\end{proof}

\begin{proposition}
\label{p-dec simple positive}
 The simple positive parity problem is PSPACE-complete. 
\end{proposition}

A corollary of the proofs of Propositions \ref{p-dec simple almost} and \ref{p-dec simple positive} is that if the 
simple almost (resp. positive) parity problem is satisfied by a word, then it is in fact satisfied also by a lasso shape word.

\begin{proposition}\label{p-undec limit buchi}
The simple limit B\"uchi and coB\"uchi problems are undecidable.
\end{proposition}

From the propositions of this section we obtain the following theorem.
In the theorem below the PSPACE-completeness of the limit safety problem 
follows as for safety conditions the limit and almost problem coincides.

\begin{theorem}
The simple almost and positive problems are PSPACE-complete for parity 
conditions.
The simple limit problem is PSPACE-complete for safety conditions, and 
the simple limit problem is undecidable for reachability, 
B\"uchi, coB\"uchi and parity conditions.
\end{theorem}

\subsection{Decidable problems for lasso shape words}\label{s-simple-words}

In this sub-section we consider the decision problems where, 
instead of restricting the probabilistic automata, we restrict the 
set of input words to lasso shape words.
First, the processes induced by such words are 
simple:

\begin{proposition}\label{p-LS-ind-simple}
 Let $\mathcal{A}$ be a PA, let $w$ be a lasso shape word, and let $\alpha\in\Delta(Q)$. Then the process induced by $w$ and $\alpha$ on $Q$ is simple.
\end{proposition}

\begin{corollary}
 Let $\mathcal{M}$ be a finite state machine. Then for any $w\in\Sigma^\omega$ generated by $\mathcal{M}$, the process induced by $w$ and $\alpha$ on $Q$ is simple.
\end{corollary}



The results of this section along with the results of the previous sub-section give us the following theorem.
\begin{theorem}\label{thrm:lasso}
Given a probabilistic automaton with parity acceptance condition,
the question whether there is lasso shape word that is accepted with 
probability~1 (or positive probability) is PSPACE-complete. 
\end{theorem}

\section{Structurally Simple Automata}\label{s-simple autom}

In this section we introduce the class of structurally simple automata, which is a structurally defined 
subclass of probabilistic automata on which every words induce 
a simple process. We show that the problems associated to this class of automata are decidable (the almost and positive problems are
PSPACE-complete and limit problem is in EXPSPACE).
We then show that this subclass of simple automata is closed under union and intersection,
and finally show that structurally simple automata strictly generalizes HPA and $\#$-acyclic
automata.


\subsection{Simple automata and structural characterization}\label{ss-SA-struct}

\begin{definition}[Simple Automata]
A probabilistic automaton is \emph{simple} if for all $w\in\Sigma^w$, the process $\lbrace \mu^w_n\rbrace_{n\in\mathbb{N}}$ induced on its state space by $w$ is simple. 
\end{definition}

In \cite{gimbert2010probabilistic}, given $S\subseteq Q$ and $a\in\Sigma$, the authors define the set $S\cdot a$ as the support of $\delta(S,a)$, and in the case where $S\cdot a=S$, the set $S\cdot a^\#$ 
as the set of states which are recurrent for the homogeneous Markov chain induced on $S$ by the transition matrix $M_a$. Next, they define the \emph{support graph} $\mathcal{G}_\mathcal{A}$ of the automaton $\mathcal{A}$ as the graph whose nodes are the subsets of $Q$, and such that, given $S,T\subseteq Q$, the couple $(S,T)$ is an edge in $\mathcal{G}_\mathcal{A}$ if there exists $a\in\Sigma$ such that $S\cdot a=T$ or $S\cdot a=S$ and $S\cdot a^\#=T$. They present the class of $\#$-acyclic automata as the class of probabilistic automata whose support graph is acyclic. 
\begin{definition}[\cite{gimbert2010probabilistic}]\label{d-acyclic PA}
 A probabilistic automaton $\mathcal{A}$ is $\#$-acyclic if $\mathcal{G}_\mathcal{A}$ is acyclic.
\end{definition}

\noindent We now present a natural generalization of this approach. Given $S\subseteq Q$ and a finite word $\rho\in\Sigma^*$, let $S\cdot\rho=\mathrm{Supp}(\delta(S,\rho))$. 
If $S\cdot\rho=S$, we define $S\cdot\rho^\#$ as the set of states which are recurrent for the homogeneous Markov chain induced on $S$ by $\rho$ 
(i.e. by the transition matrix $\lbrace\delta(q,\rho)(q')\rbrace_{q,q'\in S}$). 

\begin{example}$\ $


\begin{multicols}{2}
    \begin{enumerate}[]
      \item 

Consider the following probabilistic automaton $\mathcal{A}$, with state space $Q=\lbrace s,t,u\rbrace$.\\
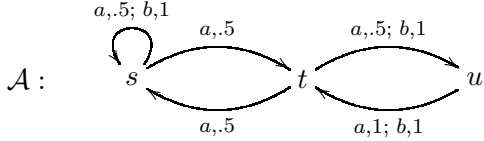
\xymatrix{
 &\mathcal{A}:& s\ar@(ur,ul)[]_{a,.5;\ b,1}\ar@/^1pc/[rr]^{a,.5}	&&	t\ar@/^1pc/[ll]^{a,.5}\ar@/^1pc/[rr]^{a,.5;\ b,1}	&&	u\ar@/^1pc/[ll]^{a,1;\ b,1}\\
}

      \item 

We have $Q\cdot a=Q\cdot a^\#=Q$, and $Q\cdot b=Q\cdot b^\#=Q$. 
However, $Q\cdot(ab)^\#=\lbrace u\rbrace$. 

    \end{enumerate}
  \end{multicols}
\end{example}




Given a probabilistic automaton $\mathcal{A}$, an \emph{execution tree} is given by an initial distribution $\alpha\in\Delta(Q)$, or a set of states $A\subseteq Q$, and a finite or infinite word $\rho$. We use the term execution tree informally for the set of execution runs on $\mathcal{A}$ which can be probabilistically generated when the system is initiated in one of the states of $\mathrm{Supp}(\alpha)$ (or $A$), and when the word $\rho$ is taken as input. 

\begin{definition}[$\#$-reductions]
A \emph{$\#$-reduction} is a tuple $(A,B,\rho)$ where $A,B\subseteq Q$ and $\rho\in\Sigma^*$ are such that: 
(i)~$A\not=\emptyset$, (ii)~$B\not=\emptyset$, 
(iii)~$A\cap B=\emptyset$, (iv)~$(A\cup B)\cdot\rho=A\cup B$, and (v)~$(A\cup B)\cdot\rho^\#=B$.
\end{definition}
For simplicity, we may use the term $\#$-reduction for a couple $(A,\rho)$ where $A\subseteq Q$ and $\rho\in\Sigma^*$ are such that $A\cdot\rho=A$ and $A\cdot\rho^\#\not=A$.

\begin{definition}
An execution tree $(\alpha,\rho)$ is said to be \emph{chain recurrent} for a probabilistic automaton $\mathcal{A}$ if it does not contain a $\#$-reduction. 
That is, for all $\rho_1,\rho_2\in\Sigma^*$ such that $\rho_1\cdot\rho_2$ is a prefix of $\rho$, $(\delta(\alpha,\rho_1),\rho_2)$ is not a $\#$-reduction. 
We write $\CRec(\alpha)$ for the set of $\rho\in\Sigma^*$ such that $(\alpha,\rho)$ is a chain recurrent execution tree for $\mathcal{A}$.
\end{definition}


The following \emph{key lemma} shows that for any probabilistic automaton $\mathcal{A}$ there exists a constant $\gamma(\mathcal{A})>0$ such that the probability 
to reach any state on a chain recurrent execution tree is either $0$ or greater than $\gamma(\mathcal{A})$. 
Given a probabilistic automaton $\mathcal{A}$, let $\epsilon(\mathcal{A})$ be the smallest non zero probability which appears among the $\delta(q,a)(q')$, where $q,q'\in Q$ and $a\in\Sigma$.

\begin{lemma}\label{l-borne mots rec}
Let $\mathcal{A}$ be a probabilistic automaton. For all $q\in Q$, all $\rho\in\CRec(q)$ and all $q'\in\mathrm{Supp}(\delta(q,\rho))$ we have $\delta(q,\rho)(q')\geq\epsilon^{2^{2\cdot|Q|}}$ where $\epsilon=\epsilon(\mathcal{A})$.
\end{lemma}

\begin{definition}[Structurally simple automata]
An automaton $\mathcal{A}$ is \emph{structurally simple} if for all $\rho\in\Sigma^*$ and $C\subseteq Q$, if $D\subseteq C$ is minimal among the $D\subseteq Q$ such that $C\stackrel{\#-\rho}{\rightarrow}D$, 
we have that $D,\rho$ is chain recurrent. Here $\#-\rho$ intuitively
denotes an iterated $\#$-reachability with the word $\rho$ (details in Section \ref{s-ext-supp} of the appendix).
\end{definition}




We now prove that all the structurally simple automata are simple.
We show that on a structurally simple automaton, given $w\in\Sigma^\omega$, 
the associated execution tree can be decomposed as a sequence of a bounded number of chain recurrent execution trees. The key Lemma \ref{l-borne mots rec} is then used to bound the probabilities which appear.

\begin{lemma}\label{l-ult pos impl simple}
Let $\lbrace\mu^w_n\rbrace_{n\in\mathbb{N}}$ be the process generated by a word $w=a_1,a_2,...\in\Sigma^\omega$ on a probabilistic automaton. Given $n\geq1$ recall that $w[1..n]=a_1,...,a_n$.
Suppose that there exists $\gamma>0$ and $N\geq0$ such that for all $n\geq N$ and all $q\in\mathrm{Supp}(\delta(\alpha,w[1..n]))$ we have 
$\delta(\alpha,w[1..n])(q)>\gamma$. Then the process is simple.
\end{lemma}

We introduce the notion of \emph{sequence of recurrent execution trees} in order to represent a process which may not be chain recurrent, 
but which can be decomposed as a sequence of a finite number of chain recurrent execution trees. The \emph{length} of the sequence measures the number of steps 
which do not belong to a chain recurrent subsequence, and will be useful to bound the probabilities which appear. 
Lemma~\ref{l-exec trees impl simple} uses the key Lemma \ref{l-borne mots rec}.

\begin{definition}[Sequence of recurrent execution trees]
A \emph{sequence of recurrent execution trees} is a finite sequence $(\alpha_1,\rho_1),\rho_1',(\alpha_2,\rho_2),\rho_2',...(\alpha_k,\rho_k)$ such that:
\begin{compactitem}
	\item $\rho_k\in\Sigma^\omega$, and for $i\in[1;k-1]$ we have $\rho_i,\rho_i'\in\Sigma^*$
	\item For all $i\in[2;k]$ we have $\mathrm{Supp}(\alpha_i)\subseteq \mathrm{Supp}(\delta(\alpha_{i-1},\rho_{i-1}\cdot\rho_{i-1}'))$
	\item All the execution trees $(\alpha_i,\rho_i)$ are chain recurrent
\end{compactitem}
The \emph{length} of the sequence is defined as $\sum_{i=1}^{k-1}|\rho_i'|$.
\end{definition}

Given an execution tree $(\alpha,w)$, a subsequence of recurrent execution trees of $(\alpha,w)$ is a sequence of recurrent execution trees $(\alpha_1,\rho_1),\rho_1',(\alpha_2,\rho_2),\rho_2',...(\alpha_k,\rho_k)$ such that $\alpha=\alpha_1$ and $w=\rho_1\cdot\rho_1'\cdot\rho_2\cdot\rho_2'\ldots\cdot\rho_k$.

\begin{lemma}\label{l-exec trees impl simple}
Let $\mathcal{A}$ be a probabilistic automaton. Suppose that there exists $K\in\mathbb{N}$ such that for all execution trees $(\alpha,\rho)$, there exists a subsequence of recurrent execution trees of length at most $K$. 
Then $\mathcal{A}$ is simple.
\end{lemma}

\begin{lemma}\label{l-s simple impl exec trees}
Suppose that $\mathcal{A}$ is structurally simple. Then for all execution trees $(\alpha,w)$, there exists a subsequence of recurrent execution trees of length at most $2^{2\cdot|Q|}$.
\end{lemma}

The following follows from Lemma \ref{l-exec trees impl simple} and \ref{l-s simple impl exec trees}.
\begin{theorem}\label{t-struct smple is simple}
All structurally simple automata are simple.
\end{theorem}

\subsection{Decision problems for structurally simple automata}\label{ss-dec-pbs-ssPA}

For the following of this sub-section, $\mathcal{A}$ is a structurally simple automaton with state space $Q$ and initial distribution $\alpha$.
We consider the complexity of the decision problems related to infinite words on structurlly simple PAs. 
The upper bound on the complexity in Theorem~\ref{p:decib_pb_on_simple} follows from the results of Section 
\ref{s-decidable pb}, since the process induced on a simple PA by a word $w\in\Sigma$ is always simple. 
The lower bound follows from the fact that the PA used for the lower bound of Section \ref{s-decidable pb} is structurally 
simple.

\begin{theorem}\label{p:decib_pb_on_simple}
The almost and positive problems are PSPACE-complete for parity conditions on structurally simple PAs.
\end{theorem}

In Proposition 6 of \cite{gimbert2010probabilistic}, the authors show that if $F\subseteq Q$ is reachable from a state $q_0$ in the support graph of $\mathcal{A}$, 
then it is limit reachable from $q_0$ in $\mathcal{A}$. A generalization of this result to the extended support graph gives half of Proposition \ref{p-simple-equiv-reach} (details with complete proof in appendix). 
Theorem \ref{t-limit-reach-PSPACE} follows from Proposition \ref{p-simple-equiv-reach} and Lemma \ref{l-algo-expaspace} 
(details in Appendix \ref{s-ext-supp} for the definition of the extended support graph). 
Theorem~\ref{prop:simple-dec} shows that the limit parity problem is also decidable for simple automata.

\begin{proposition}\label{p-simple-equiv-reach}
 Let $\mathcal{A}$ be a structurally simple automaton, and let $F\subseteq Q$. Then 
\textbf{(1)} $F$ is reachable from $\mathrm{Supp}(\alpha)$ in the extended support graph of $\mathcal{A}$ iff \textbf{(2)} it is limit reachable from $\alpha$ in $\mathcal{A}$.
\end{proposition}

\begin{theorem}\label{t-limit-reach-PSPACE}
The limit problem is in EXPSPACE for reachability conditions on structurally 
simple PAs.
\end{theorem}

\begin{theorem}
\label{prop:simple-dec}
The limit problem is in EXPSPACE for parity conditions on structurally simple PAs.
\end{theorem}

The following theorem establishes the decidability of the problem that given 
a probabilistic automaton whether the automaton is structurally simple.

\begin{theorem}\label{t-decid-ssimple}
We can decide in EXPSPACE whether a given probabilistic automaton is structurally simple or not.
\end{theorem}

\vspace{-1em}

\subsection{Closure properties for Structurally Simple Automata}\label{ss-clos-props}

Given $\mathcal{A}_1=(S_1,\Sigma,\delta_1,\alpha_1)$ and $\mathcal{A}_2=(S_2,\Sigma,\delta_2,\alpha_2)$ two structurally simple automata on the same alphabet $\Sigma$, 
the construction of the Cartesian product automaton $\mathcal{A}_1\Join\mathcal{A}_2$ is standard. We detail this construction in appendix, along with the proof of the following proposition.

\begin{proposition}\label{p-product simple}
 Let $\mathcal{A}_1$ and $\mathcal{A}_2$ be two structurally simple automata. Then $\mathcal{A}=\mathcal{A}_1\Join\mathcal{A}_2$ is structurally simple.
\end{proposition}

We prove that the classes of languages recognized by structurally simple automata under various semantics (positive parity, almost parity) are robust. 
This property relies on the fact that one can construct the intersection or the union of two parity (non-probabilistic) automata using only product constructions and change in the semantics (going from parity to Streett or Rabin, and back to parity; see~\cite{Thomas97,baier2005recognizing} for details of Rabin and 
Streett conditions and the translations). By Proposition \ref{p-product simple}, such transformations keep the automata simple.

\begin{theorem}\label{t-class-robust}
 The class of languages recognized by structurally simple automata under positive (resp. almost) semantics and parity condition is closed under union and intersection.
\end{theorem}

\subsection{Subclasses of Simple Automata}\label{s-subclasses}

In this section we show that both $\#$-acyclic automata (recall Definition \ref{d-acyclic PA}) and 
hierarchical probabilistic automata are strict subclasses of 
simple automata.

\begin{proposition}\label{p-extends_acycl}
The class of structurally simple automata strictly subsumes the class of $\#$-acyclic automata.
\end{proposition}

Another restriction of Probabilistic Automata which has been considered is the model of \emph{Hierarchical PAs}, presented first in \cite{chadha2009power}. Intuitively, a hierarchical PA is a probabilistic automaton on which a \emph{rank function} must increase on every runs. This condition imposes that the induced processes are ultimately deterministic with probability one.

\begin{definition}[\cite{chadha2009power}]
Given $k\in\mathbb{N}$, a PA $\mathcal{B}=(Q,q_s,Q,\delta)$ over an alphabet $\Sigma$ is said to be a \emph{$k$-level hierarchical PA} ($k$-HPA) if there is a function $\mathrm{rk}:Q\rightarrow\lbrace0,1,...,k\rbrace$ such that the following holds:
\begin{center}
 Given $j\in\lbrace0,1,...,k\rbrace$, let $Q_j=\lbrace q\in Q\ |\ \mathrm{rk}(q)=j\rbrace$. For every $q\in Q$ and $a\in\Sigma$, if $j_0=\mathrm{rk}(q)$ then $post(q,a)\subseteq \cup_{j_0\leq l\leq k}Q_l$ and $|post(q,q)\cap Q_{j_0}|\leq1$.
\end{center}

\end{definition}

\begin{proposition}\label{p-extends_hierch}
The class of structurally simple automata strictly subsumes the class of Hierarchical PAs.
\end{proposition}

It follows that our decidability results for structurally simple 
PAs both unifies and generalizes the decidability results previously known for 
$\#$-acyclic (for limit reachability) and hierarchical PA (for almost and positive B\"uchi).

\section{Conclusion} 
In this work we have used a very 
general result from stochastic processes, namely the decomposition-separation
theorem, to identify simple structure of tail $\sigma$-fields,
and used them to define simple processes on probabilistic automata. 
We showed that under the restriction of simple processes the almost 
and positive decision problems are decidable for all 
parity conditions.
We then characterized structurally a subclass of the class of 
simple automata on which every process is simple. 
We showed that this class is decidable, robust, and that it generalizes 
the previous known subclasses of probabilistic 
automata for which the decision problems were decidable. 
Our techniques also show that for lasso shape words the 
almost and positive decision 
problems are decidable for all probabilistic automata.
We believe that our techniques will be useful in future research 
for other decidability results related to probabilistic automata 
and more general probabilistic models (such as partially observable
Markov decision processes and partial-observation stochastic games).

\newpage
\section*{Appendix}
\appendix

\newcommand{\Rec}{\mathrm{Rec}}
\newcommand{\Comp}{\mathrm{Comp}}
\newcommand{\leftc}{\mathsf{left}}
\newcommand{\rightc}{\mathsf{right}}

\newcommand{\LG}{\mathsf{LG}}
\newcommand{\org}{\mathrm{org}}
\newcommand{\dest}{\mathrm{dest}}

\section{Details of Section \ref{s-simple process}}

\smallskip\noindent{\bf Details of Proposition \ref{p-tau atomics}.} 

\begin{proof} \emph{(of Proposition~\ref{p-tau atomics}).}We present the proof of all three parts below.

\smallskip\noindent{\bf Assertion~1.}
Let $i\in\set{1,2,\ldots,c}$. We prove that $\tau^i_\infty$ belongs to $\mathcal{F}_\infty$. 
We first note that for all $N_0\in\mathbb{N}$ and $N\geq N_0$, by definition of $\mathcal{F}_N$, we have $\cap_{n\geq N}\tau_i^n\in\mathcal{F}_{N_0}$. 
Next, we note that $\lbrace\cap_{n\geq N}\tau_i^n\rbrace_{N\in\mathbb{N}}$ is an increasing sequence of sets of runs, and that the first point of 
Theorem \ref{t-decomp} implies $\tau_\infty^i=\mathrm{lim}_{N\rightarrow\infty}\cap_{n\geq N}\tau^i_n$. For all $N_0\in\mathbb{N}$, we have $\mathcal{F}_{N_0}$ is a $\sigma$-field, 
hence the limit of an increasing sequences of sets in $\mathcal{F}_{N_0}$ also belong to $\mathcal{F}_{N_0}$. We get that for all $N_0\in\mathbb{N}$, we have $\tau_\infty^i\in\mathcal{F}_{N_0}$, hence the result.

\smallskip\noindent{\bf Assertion~2.}
We prove that $\tau^i_\infty$ is atomic by using Proposition 2.1. of \cite{cohn1989products}, 
which states the following result:
\begin{compactitem}
\item 
For any set $\Gamma$ in $\mathcal{F}_\infty$, there exists a sequence $L_n$ of subsets of $Q$ such that, $\mathbb{P}^{w}$-almost surely, we have $\mathrm{lim}_{n\rightarrow\infty}\lbrace r\in\Omega\ s.t.\ r(n)\in L_n\rbrace=\Gamma$.
\end{compactitem}
Here, ``$\mathrm{lim}_{n\rightarrow\infty}\lbrace r\in\Omega\ s.t.\ r(n)\in L_n\rbrace=\Gamma$ almost surely'' means that the $\mathbb{P}^{w}$-measure of the set of states 
on which the characteristic functions of the sets $\lbrace r\in\Omega\ s.t.\ r(n)\in L_n\rbrace$ and $\Gamma$ goes to zero as $n$ goes to infinity. For sake of completeness, we prove this fact 
using the Martingale Convergence Theorem, as in \cite{cohn1989products} 
(see for instance \cite{kemeny1976denumerable} for a presentation of the Martingale Convergence Theorem and the Levy's Law).

Given $n\in\mathbb{N}$, let $\sigma(X^w_0,X^w_1,\ldots,X^w_n)$ be the $\sigma$-field generated by $X^w_i,i\in\set{0,\ldots,n}$. 
Since $\Gamma$ belongs to $\mathcal{F}_\infty=\cap_{n\in\mathbb{N}}\mathcal{F}_n$, the Levy's Law implies that, 
$\mathbb{P}^w$ almost surely, $\mathrm{lim}_{n\rightarrow\infty}\mathbb{P}(\Gamma|\sigma(X^w_0,X^w_1,\ldots,X^w_n))=1_{\Gamma}$, 
where $1_\Gamma$ is the characteristic function of $\Gamma$. Since $\lbrace X_n,\ n\geq0\rbrace$ is Markovian, 
we know that for all $n$ we have $\mathbb{P}(\Gamma|\sigma(X^w_0,X^w_1,\ldots,X^w_n))=\mathbb{P}(\Gamma|\sigma(X^w_n))$. 
Let $0<\lambda<1$, and given $n\in\mathbb{N}$ let $L_n=\lbrace q\in Q\ |\ \mathbb{P}(\Gamma|X^w_n=q)>\lambda\rbrace$. 
Then, $\mathbb{P}^w$ almost surely, we have $\mathrm{lim}_{n\rightarrow\infty}\lbrace X_n\in L_n\rbrace=\Gamma$, which proves the preliminary result.

Now, let $A\in\tau^i_\infty$. By hypothesis, $\mathbb{P}^{w}[A]>0$. Suppose by contradiction that $0< \mathbb{P}^{w}[A]< \mathbb{P}^{w}(\tau^i_\infty)$. 
Let $B=\tau^i_\infty\setminus A$. We have $A,B\in\mathcal{F}_\infty$, hence there exist $L_n,L_n',\ n\in\mathbb{N}$ two sequences of sets such that $\mathrm{lim}_{n\rightarrow\infty}\lbrace r\in\Omega\ |\ r(n)\in L_n\rbrace=A$ almost surely and $\mathrm{lim}_{n\rightarrow\infty}\lbrace r\in\Omega\ |\ r(n)\in L_n'\rbrace=B$ almost surely. Let $N$ be large enough, and let $q\in L_N,\ q'\in L_N'$ be such that :
\[\mathbb{P}[r\in A\ |\ r(N)=q]>1-\frac{1}{4\cdot|Q|^2}; 
\]
and
\[
\mathbb{P}[r\in B\ |\ r(N)=q']>1-\frac{1}{4\cdot|Q|^2}.
\]
We prove that this contradicts the second point of Theorem \ref{t-decomp}: first, by the Pigeon Hole Principle, there exists a sequence $q_n,\ n\geq N$ of states in $L_n'$ such that
\[\mathrm{lim}_n\mathbb{P}[r(n)=q_n |\ r(N)=q']>\dfrac{1}{2\cdot|Q|}.\]
Moreover, by the second point of Theorem \ref{t-decomp} we know that 
\[\mathrm{lim}_{n\rightarrow\infty}\dfrac{\mathbb{P}[X^w_n=q_n\ |\ X^w_n\in J^k_n\wedge X^w_m=q]}{\mathbb{P}[X^w_n=q_n\ |\ X^w_n\in J^k_n\wedge X^w_m=q']}=1\]
Thus, for $n$ large enough, $\mathbb{P}[r(n)=q_n |\ r(N)=q]>\dfrac{1}{4\cdot|Q|}$. Hence, for $n$ large enough, $\mathbb{P}[r\not\in A |\ r(N)=q]>\dfrac{1}{4\cdot|Q|^2}$. This is a contradiction.

\smallskip\noindent{\bf Assertion~3.}
The fact that, $\mathbb{P}^{w}(\bigcup_{i=1}^c\tau^i_\infty)=1$, is a consequence of the first point of Theorem \ref{t-decomp}: 
with probability one, after a finite number of steps, a run belongs to one of the $J^i$ and never leaves it.
\end{proof}

\begin{comment}
\begin{proof} \emph{(of Corollary~\ref{c-tau fini}).} 
 Direct by the second point of Proposition \ref{p-separation simple jets}, since we know that with probability one, after a finite number of steps, a run enters in one of the $J^i, i\in\set{1,2,\ldots,c}$, and never leaves it.
\end{proof}
\end{comment}

\smallskip\noindent{\bf Details of Proposition \ref{p-separation simple jets}.} 
We prove Proposition \ref{p-separation simple jets}.

\begin{proof}
Let $J^0,...,J^c$ be a decomposition of $Q^\omega$ into jets, as in Theorem \ref{t-decomp}. Let $\lambda>0$ be the threshold given by the definition of a simple process, for the process $\lbrace \mu^w_n\rbrace_{n\in\mathbb{N}}$. For all $i\in\set{1,2,\ldots,c}$ and $n\in\mathbb{N}$, let 
\[\widehat{J}^i_n=\lbrace q\in J^i_n\ s.t.\ \mu^w_n(q)>\lambda\rbrace\]
For all $n\in\mathbb{N}$, let $H^0_n=J^0_n\cup\bigcup_{i=1}^c(J^i_n\setminus\widehat{J}^i_n)$, and let $H^i_n=\widehat{J}^i_n$ for $i\in\set{1,2,\ldots,c}$. We claim that the decomposition of $Q^\omega$ into jets $H=(H^0,...,H^c)$ satisfies the conditions of the proposition.

The first point of the Proposition follows from the definition of the $\widehat{J}^i_n$. The third point follows from the fact that the process is simple: the probability of the set of states whose measure is less than $\lambda$ goes to zero. We prove now the second point.

Suppose that there exists no $N\in\mathbb{N}$ such that the property is satisfied for the jet decomposition $H$. Then, there exists $i\in\set{1,2,\ldots,c}$ such that for all $N\in\mathbb{N}$, there exist $n_2>n_1\geq N$ such that $\delta(\widehat{J}^i_{n_1},w^{n_2}_{n_1})\not\subseteq \widehat{J}^i_{n_2}$.

We write $w=a_0,a_1,...$. Since $Q$ is finite, there exist $i\not=j$ in $\set{1,2,\ldots,c}$ and $q,q'\in Q$ such that for an infinite number of $n\in\mathbb{N}$ we have $q\in\widehat{J}^i_n,\ q'\in\widehat{J}^j_n$, and $\delta(q,a_n)(q')>0$. Since for $n$ large enough we have $\mu^w_n(q)>\lambda$ for all $q\in \widehat{J}^i_n$, this implies that the probability of the set of runs which move from jet $J^i$ to jet $J^j$ infinitely often is at least $\epsilon\cdot\lambda$, where $\epsilon$ is the least non zero probability which appears among the transition probabilities given by the $M_a$, for $a\in\Sigma$. This implies that the probability of the set of runs which stay inside one of the $J^i,\ i\in\set{1,2,\ldots,c}$ for ever after a finite number of steps cannot be equal to one. This contradicts the definition of the decomposition $J^0,...,J^c$.

For the fourth point, the fact that the jets are mixing follows directly from the Theorem \ref{t-decomp}, and the fact that a run does not leave $\widehat{J}^i$ once it has entered it after time $N$.
\end{proof}

\section{Details of Section \ref{s-decidable pb}}

\smallskip\noindent{\bf Details of Proposition \ref{prop:dec-undec-simple}.} 
We prove Proposition \ref{prop:dec-undec-simple}.

\begin{proof}
By \cite{baier2005recognizing} and \cite{chatterjee2010}, the almost (resp. positive) safety and reachability problems are decidable for the general 
class of probabilistic automata, as well as the limit safety problem. The results of \cite{baier2005recognizing} and \cite{chatterjee2010} 
show that if one of the problems is satisfiable, it is satisfiable by a lasso shape word, and hence the simple version of the problem is satisfiable (by the results of our Section \ref{s-simple-words}). 
As a consequence, we can use this result to get the decidability of the problems when we 
restrict to simple processes. The PSPACE-completeness follows from the results of \cite{chadha2009power}.


The undecidability of the limit reachability problem comes from the results of \cite{gimbert2010probabilistic} and \cite{chatterjee2010}, which show that it is undecidable 
for the general class of probabilistic automata, and from the following fact: Given a PA with state space $Q$, accepting states $F\subseteq Q$ and $\epsilon\in]0;1[$, 
if there exists $w\in\Sigma^\omega$ such that $\mathbb{P}^w[\lbrace r\ |\ r\in \Reach(F)\rbrace]>1-\epsilon$, then there exists $w'\in\Sigma^\omega$ such that 
$\mathbb{P}^{w'}[\lbrace r\ |\ r\in \Reach(F)\rbrace]>1-2\cdot\epsilon$ and the process induced by $w'$ is simple. For this we just have to consider any lasso shape word 
$w=\rho_1\cdot\rho_2^\omega$ whose prefix word $\rho_1$ satisfies the $1-2\cdot\epsilon$ reachability condition. In Section \ref{s-simple-words}, we see that 
the process induced by a lasso-shape word on an automaton is always simple, which concludes the proof.
\end{proof}

\smallskip\noindent{\bf Details of Proposition \ref{p-dec simple almost}.}

\begin{proof}

The proof is in three parts: first we present an equivalent formulation of the problem. 
Then we show that the equivalent formulation gives a problem which we can solve in PSPACE. 
Finally we give the PSPACE lower bound.

\smallskip\noindent{\bf Equivalent formulation.}
In the following, $p:Q\rightarrow\mathbb{N}$ is a parity function on $Q$, and $\Phi=\Parity(p)$. We prove that: \textbf{(1)} There exists $w\in\Sigma^\omega$ such that the induced process is simple and $\mathbb{P}^w_\mathcal{A}(\Phi)=1$ if and only if \textbf{(2)} There exists $G\subseteq Q$ and $\rho_1,\rho_2\in\Sigma^*$ such that $G=\delta(\alpha,\rho_1)$, $\delta(G,\rho_2)\subseteq G$, and the runs on the Markov chain induced by $(G,\rho_2)$ satisfy $\Phi$ with probability one.
We then show that the properties can be verified in PSPACE and also present a PSPACE lower bound.

We show the equivalence \textbf{(2)}$\Leftrightarrow$\textbf{(1)}. The way \textbf{(2)}$\Rightarrow$\textbf{(1)} is direct, since we will show in Section \ref{s-simple-words} that the process induced by a lasso shape word on any automaton is always simple. We prove that \textbf{(1)}$\Rightarrow$\textbf{(2)}. Let $w=a_1,...,a_i,...$ be such that the induced process is simple and $\mathbb{P}^w_\mathcal{A}(\Phi)=1$. Using Proposition \ref{p-separation simple jets}, let $J^0,J^1,...,J^m$ be the decomposition of $Q^\omega$ into jets and let $N_0\in\mathbb{N},\ \lambda>0$ be such that:
\begin{compactitem}
 \item $\forall n\geq N_0$, $\forall i\in\set{1,2,\ldots,c}$, $\forall q\in J^i_n$: $\mu_n(q)>\lambda$.
 \item $\forall i\in\set{1,2,\ldots,c}$, for all $n_2>n_1\geq N_0$, we have $\delta(J^i_{n_1},w^{n_2}_{n_1+1})\subseteq J^i_{n_2}$.
 \item $\mu^w_n(J^0_n)\rightarrow_{n\rightarrow\infty}0$
 \item Each jet $J^i,\ i\in\set{1,2,\ldots,c}$ is mixing.
\end{compactitem}
Without loss on generality, since $Q$ is finite, taking $N_0$ large enough, we can assume that the vector of sets of states $(J_{N_0}^0,...,J_{N_0}^c)$ appears infinitely often in the sequence $\lbrace(J_{n}^0,...,J_{n}^c)\rbrace_{n\in\mathbb{N}}$. As well, without loss on generality, we can assume that for all $n\geq N_0$ and all $i\in\set{1,2,\ldots,c}$, all the states in $J^i_n$ appear infinitely often among the sets $J^i_m$, for $m\geq N_0$. Let $i\in\set{1,2,\ldots,c}$. Given $q\in Q$, let 
\[\Phi_q=\lbrace r\in\Omega\ |\ q\in\mathrm{Inf}(r)\ \mathrm{and}\ p(q)=\min_{q'\in\mathrm{Inf}(r)}p(q')\rbrace\]
Clearly, for all $q\in Q$, $\Phi_q\in\mathcal{F}_\infty$. Since $Q$ is finite, there exists $q_i\in Q$ such that $\mathbb{P}(\tau^i_\infty\cap \Phi_{q_i})>0$. By Proposition \ref{p-tau atomics}, $\tau_\infty^i$ is atomic, hence $\tau_\infty^i\subseteq \Phi_{q_i}$. Since the runs of the process satisfy the parity condition with probability one, $p(q_i)$ must be even. Moreover, for all $n\geq N_0$ and all $q\in J^i_n$, we must have $p(q)\geq p(q_i)$. Indeed, such a $q$ appears an infinite number of times in the sequence $J^i_n$, by hypothesis, and always with probability at least $\lambda$. 

Since $\tau_\infty^i\subseteq \Phi_{q_i}$, there exists $m_i\in\mathbb{N}$ such that for all $q\in J_i^{N_0}$, there exists $m<m_i$ such that $\delta(q,w[N_0+1..m])(q_i)>0$. We define $m=\max_{i\in\set{1,2,\ldots,c}}\ m_i$, and $m'\geq m$ such that 
\[(J_{N_0}^0,...,J_{N_0}^c)=(J_{N_0+m'}^0,...,J_{N_0+m'}^c)\]
Taking $\rho_1\!\!=\!\!w[0..N_0]$ and $\rho_2\!=\!w[N_0+1..N_0+m']$ completes the proof. Indeed, when starting from the initial distribution, after reading $\rho_1$, we arrive by construction in one of the sets $J^i_{N_0}$, with $i\in\set{0,\ldots,c}$. Starting from this state $q$, if the word $\rho_2$ is taken as input, we go to set $J^i_{N_0+m'}$ with probability one, visit $q_i$ with positive probability, and do not visit any state with probability smaller that $p(q_i)$. This implies that when starting from $q$ and reading $\rho_2^\omega$, we visit $q_i$ with probability one, hence the result.

Now, we argue the PSPACE upper and lower bounds. 

\smallskip\noindent{\bf PSPACE upper bound.}
First, we show that we can verify the second property in NPSPACE, hence in PSPACE. The proof is in two steps. In a first step, we show that we can decide in NPSPACE whether, given $G\subseteq Q$, there exists $\rho_1\in\Sigma^*$ such that $G=\delta(\alpha,\rho_1)$. For this notice that, given $G\subseteq Q$, if there exists $\rho_1\in\Sigma^*$ such that $G=\delta(\alpha,\rho_1)$, then there exists $\rho_1'\in\Sigma^*$ such that $G=\delta(\alpha,\rho_1')$ and $|\rho_1'|\leq2^{|Q|}$. Thus, we can restrict the search to words $\rho_1$ of length at most $2^{|Q|}$. By guessing the letters $a_1,a_2,\ldots$ of $\rho_1$ one by one, and by keeping in memory the set $A_i=\delta(\alpha,a_1,\ldots,a_i)$ at each step, we can check at each step whether $A_i=G$, and thus we can decide whether there exists such a $\rho_1$ in NPSPACE.

In a second step, we show that, given $G\subseteq Q$, we can decide in NPSPACE whether there exists $\rho_2\in\Sigma^*$ such that the runs on the periodic non-homogeneous Markov chain induced by $(G,\rho_2)$ satisfy $\Phi$ with probability one. For this, we refine the previous argument. Notice that this is equivalent to find $\rho_2=a_1,\ldots,a_k\in\Sigma^*$ and $A,B\subseteq Q$ such that:
\begin{compactitem}
 \item $\rho_2$ has length at most $2^{2\cdot|Q|}$
 \item $\delta(G,\rho_2)\subseteq G$
 \item $A,B$ partition $G$
 \item $A$ is the set of recurrent states for the homogeneous Markov chain induced by $\rho_2$ on $G$
 \item $B$ is the set of transient states for the homogeneous Markov chain induced by $\rho_2$ on $G$
 \item For all $q_0\in A$, for all the finite runs $q_0,a_1,q_1,a_2,q_2,\ldots,a_k$ generated with positive probability when initiated on $q$ and when reading $\rho_2$, the minimal value of the $p(q_i),\ i\in\set{0,k-1}$ is even.
\end{compactitem}
This can be checked in NPSPACE. Indeed, we can guess $A,B$, and the letters of $\rho_2$ one by one, and at each step keep in memory the following sets:
\begin{compactitem}
 \item The set of states visited at time $i$, i.e. $E_i=\delta(A\cup B,a_1,\ldots,a_i)$
 \item For all $q\in A$ and all $q'\in\delta(q,a_1,\ldots,a_i)$, the minimal $p$ value of the paths visited between $q$ and $q'$. Notice that this set has size at most $|Q|$.
 \item For all $q\in A\cup B$ and all $q'\in\delta(A\cup B,a_1,\ldots,a_i)$, a boolean value $v_i(q,q')$ which is equal to one if there exists a path between $q$ and $q'$ between the first step and step $i$, and which is null if not.
\end{compactitem}
At the end, we just have to check that $E_k=G$, that the minimal $p$-values of all the paths issued from $A$ is even, that the set of states in $A$ are recurrent for the chain, and that the states in $B$ are transient. This can be done easily since we can recover the graph of the Markov chain on $G$ from the values given by $v_{|\rho_2|}$.


\smallskip\noindent{\bf PSPACE lower bound.}
We prove now that the simple almost B\"uchi problem is PSPACE-hard. For this, we reduce the problem of checking the emptiness of a finite intersection of regular languages, which is known to be PSPACE complete by \cite{kozen1977lbn}, to the \emph{simple almost B\"uchi problem}, which is a particular case of the simple almost parity problem. The size of the input of Problem \ref{p-inter reg lang} is the sum of the number of states of the automata.

\begin{problem}[Finite Intersection of Regular Languages]\label{p-inter reg lang}$\ $\\
Input: $\mathcal{A}_1,...,\mathcal{A}_l$ a family of regular deterministic automata (on finite words) on the same finite alphabet $\Sigma$. \\
Question: Do we have $\mathcal{L}(\mathcal{A}_1)\cap...\cap\mathcal{L}(\mathcal{A}_l)=\emptyset$ ?
\end{problem}

Let $\mathcal{A}_1,...,\mathcal{A}_l$ be a family of regular automata on the same finite alphabet $\Sigma$, with respective state spaces $Q_i$ and transition functions $\delta_i$ (where $\delta_i(s,a)(t)=1$ if there exist a transition from $s$ to $t$ with label $a\in\Sigma$ in $\mathcal{A}_i$). We build a probabilistic automaton $\mathcal{A}=(Q,\Sigma',\delta,\alpha,F)$ such that the simple almost $\Buchi(F)$ problem is satisfied on $\mathcal{A}$ iff $\mathcal{L}(\mathcal{A}_1)\cap...\cap\mathcal{L}(\mathcal{A}_l)\not=\emptyset$.

 Let $x$ be a new letter, not in $\Sigma$, and let $\Sigma'=\Sigma\cup\lbrace x\rbrace$. 
\begin{compactitem}
	\item $Q$ is the union of the state spaces of the $\mathcal{A}_i$, plus two extra states $s$ and $\perp$. That is $Q=\bigcup_{i=1}^l Q_i'\cup\lbrace s,\perp\rbrace$, where the $Q_i'$ are disjoint copies of the $Q_i$.
	\item The state $\perp$ is a sink: for all $a\in\Sigma'$, $\delta(\perp,a)(\perp)=1$.
	\item If $u'$ is the copy of a non accepting state $u$ of $\mathcal{A}_i$, we allow in $\mathcal{A}$ the same transitions from $u'$ as in $\mathcal{A}_i$ for $u$: if $a\in\Sigma$, let $\delta(u',a)(v')=1$ iff $v'$ is the copy of a state $v\in Q_i$ such that $\delta_i(u,a)(v)=1$. Moreover we add a transition from $u$ with label $x$: $\delta(u,x)(\perp)=1$.
	\item If $u'$ is the copy of an accepting state $u$ of $\mathcal{A}_i,i\in[1;l]$, the transitions from $u'$ in $\mathcal{A}$ are the same as in $\mathcal{A}_i$, plus an extra transition $\delta(u',x)(s)=1$.
	\item From state $s$ in $\mathcal{A}$, with uniform probability on $i\in[1;l]$, when reading $x$, the system goes to one of the copies of an initial state of the $\mathcal{A}_i$'s. 
	\item For the transitions which have not been precised, for instance if $a\in\Sigma$ is read in state $s$, the system goes with probability one to the sink $\perp$.
	\item The initial distribution $\alpha$ is the Dirac distribution on $s$.
	\item $F=\lbrace s\rbrace$
\end{compactitem}

Given $\rho\in\mathcal{L}(\mathcal{A}_1)\cap...\cap\mathcal{L}(\mathcal{A}_l)$, the input word $(x\cdot\rho\cdot x)^\omega$ satisfies clearly the simple almost $\Buchi(F)$ problem since a run visits $s$ after each occurrence of $x\cdot\rho\cdot x$ (the generated process is simple since we see in Section \ref{s-simple-words} that any process generated on a probabilistic automaton by a lasso shape word is simple).

Conversely, suppose that there exists $\rho\in\Sigma^\omega$ such that the induced process is simple and satisfies almost surely the $\Buchi(F)$ condition.

\begin{compactitem}
 \item Since the only transition from $s$ which does not goes to the sink has label $x$, the word $\rho$ must start with letter $x$. 
 \item Since with probability one the runs induced by $\rho$ visit infinitely often $s$, the letter $x$ must appear infinitely often in $\rho$. Let $\rho=x\cdot\rho'\cdot x$ where $\rho'\in\Sigma$ is non empty and does not contain the letter $x$. After reading $x\cdot\rho'\cdot x$, since the process cannot be in the sink $\perp$ with positive probability, it has to be on $s$ with probability one. This implies that $\rho'\in\mathcal{L}(\mathcal{A}_1)\cap...\cap\mathcal{L}(\mathcal{A}_l)$, hence $\mathcal{L}(\mathcal{A}_1)\cap...\cap\mathcal{L}(\mathcal{A}_l)\not=\emptyset$.
\end{compactitem}

This concludes the proof of the PSPACE completeness of our problem. We give an example of the last reduction.

\begin{example}
 Consider the following regular automata $\mathcal{A}_1$ and $\mathcal{A}_2$, and the associated probabilistic automaton $\mathcal{A}$.
\begin{figure}[!ht]
\begin{center}
\begin{picture}(50,30)(20,0)
  \gasset{Nadjust=w,Nadjustdist=2,Nh=6,Nmr=1}
  \node[Nmarks=r](A)(55,0){$1$}
  \node[Nmarks=i](B)(30,0){$2$}
    \drawloop[loopdiam=5,loopangle=110,ELpos=60](A){$a$}
    \drawedge[curvedepth=0,ELside=l](A,B){$b$}
    \drawloop[loopdiam=5,loopangle=110,ELpos=60](B){$b$}
    \drawedge[curvedepth=5,ELside=l](B,A){$a$}

  \node[Nmarks=i](C)(25,20){$3$}
  \node[Nmarks=r](D)(45,20){$4$}
  \node(E)(65,20){$5$}
    \drawloop[loopdiam=5,loopangle=110,ELpos=60](C){$b$}
    \drawedge[curvedepth=0,ELside=l](C,D){$a$}
    \drawloop[loopdiam=5,loopangle=110,ELpos=60](D){$a$}
    \drawedge[curvedepth=0,ELside=l](D,E){$b$}
    \drawloop[loopdiam=5,loopangle=110,ELpos=60](E){$a,b$}
\end{picture}
\end{center}
\caption{Automata $\mathcal{A}_1$ and $\mathcal{A}_2$}
\end{figure}
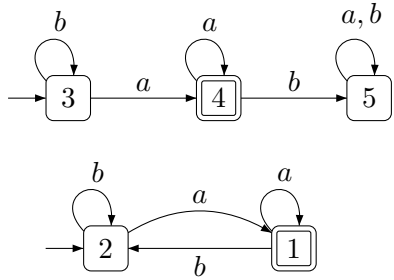


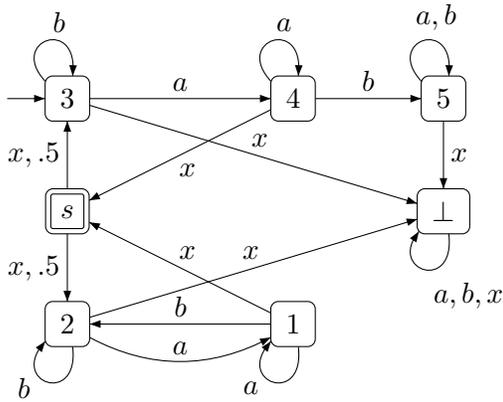
\begin{figure}[!ht]
\begin{center}
\begin{picture}(140,50)(20,-10)
  \gasset{Nadjust=w,Nadjustdist=2,Nh=6,Nmr=1}

  \node(A)(60,0){$1$}
  \node(B)(30,0){$2$}
    \drawloop[loopdiam=5,loopangle=-110,ELpos=60](A){$a$}
    \drawedge[curvedepth=0,ELside=r](A,B){$b$}
    \drawloop[loopdiam=5,loopangle=-110,ELpos=60](B){$b$}
    \drawedge[curvedepth=-5,ELside=l](B,A){$a$}

  \node[Nmarks=i](C)(30,30){$3$}
  \node(D)(60,30){$4$}
  \node(E)(80,30){$5$}
    \drawloop[loopdiam=5,loopangle=110,ELpos=60](C){$b$}
    \drawedge[curvedepth=0,ELside=l](C,D){$a$}
    \drawloop[loopdiam=5,loopangle=110,ELpos=60](D){$a$}
    \drawedge[curvedepth=0,ELside=l](D,E){$b$}
    \drawloop[loopdiam=5,loopangle=110,ELpos=60](E){$a,b$}

  \node[Nmarks=r,Nmarks=r](I)(30,15){$s$}
  \node(F)(80,15){$\perp$}

  \drawedge[curvedepth=0,ELside=r](I,B){$x,.5$}
  \drawedge[curvedepth=0,ELside=l](I,C){$x,.5$}
  \drawedge[curvedepth=0,ELside=r](A,I){$x$}
  \drawedge[curvedepth=0,ELside=l](D,I){$x$}
 
  \drawedge[curvedepth=0,ELside=l](B,F){$x$}
  \drawedge[curvedepth=0,ELside=l](C,F){$x$}
  \drawedge[curvedepth=0,ELside=l](E,F){$x$}

   \drawloop[loopdiam=5,loopangle=-110,ELpos=30](F){$a,b,x$}

\end{picture}
\end{center}
\caption{The probabilistic automaton $\mathcal{A}$}
\end{figure}
For instance, the word $b\cdot a\cdot a$ belongs to $\mathcal{L}(\mathcal{A}_1)\cap\mathcal{L}(\mathcal{A}_2)$. We get that on $\mathcal{A}$, the word $(x\cdot b\cdot a\cdot a\cdot x)^\omega$ satisfies the simple almost $\Buchi(\lbrace s\rbrace)$ problem.

\end{example}
This completes the details of the PSPACE upper and lower bound.

\end{proof}

\smallskip\noindent{\bf Details of Proposition \ref{p-dec simple positive}.} 
We prove that the simple positive parity problem is PSPACE-complete.

As for the proof of Proposition \ref{p-dec simple almost}, the proof is in three parts: first we present an equivalent formulation of the problem. 
Then we show that the equivalent formulation gives a problem which we can solve in PSPACE. 
Finally we give the PSPACE lower bound. We present only the first part of the proof in details, since the second and third parts are analogous to the proof 
of Proposition \ref{p-dec simple almost}.
 
\begin{proof}

As before, let $\Phi=\Parity(p)$ where $p:Q\rightarrow\mathbb{N}$ is a parity function. We follow a method analogous to the one for the simple almost parity problem: We prove that: \textbf{(1)} 
There exists $w\in\Sigma^\omega$ such that the induced process is simple and $\mathbb{P}^w_\mathcal{A}(\Phi)>0$ iff 
\textbf{(2)} There exists $G\subseteq Q$ and $\rho_1,\rho_2\in\Sigma^*$ such that $G\subseteq\mathrm{Supp}(\delta(\alpha,\rho_1))$, and $\delta(G,\rho_2)\subseteq G$, 
and the runs on the Markov chain induced by $(G,\rho_2)$ satisfy $\Phi$ with probability one.
That is, we reach $G$ with positive probability, and once we read $\rho_2$ from a state in $G$ we satisfy the condition almost surely.

The way \textbf{(2)}$\Rightarrow$\textbf{(1)} of the equivalence is direct. We prove that conversely, \textbf{(1)}$\Rightarrow$\textbf{(2)}. Suppose now that there exists such a $w=a_1,...,a_i,...$ . The induced processed is simple, so let $J^0,J^1,...,J^m$ be as given by Proposition \ref{p-separation simple jets}. As before, without loss on generality, since $Q$ is finite, we can also assume that the vector of sets $(J^{N_0}_0,...,J^{N_0}_c)$ appears infinitely often in the sequence $(J_{n}^0,...,J_{n}^c),\ n\in\mathbb{N}$. Moreover, we also assume that for all $n\geq N_0$, for all $i\in\set{1,2,\ldots,m}$, all the states in $J_i^n$ appears in a infinite number of the sets $J_i^m,m\geq N_0$.

Let $i\in\set{1,2,\ldots,c}$. As before, an ultimate property is either satisfied or unsatisfied with probability one by the runs $r\in\Omega$ such that $r(N_0)\in J^i_n$. Thus, we can define $q_i\in Q$ as the state with minimal value for $p$ among the states which are visited infinitely by runs in $\tau^\infty_i$ with probability one.

Since the runs of the process satisfy the parity condition with positive probability, there exists $i\in\set{1,2,\ldots,c}$ such that $p(q_i)$ is even. Moreover, for all $n\geq N_0$ and all $q\in J_i^n$, as in the previous case, we must have $p(q)\geq p(q_i)$. Finally, there exists $m_i\in\mathbb{N}$ such that for all $q\in J_i^{N_0}$, there exists $m<m_i$ such that $\delta(q,w[N_0..m])(q_i)>0$. We define $m'\geq m_i$ such that 
\[(J_{N_0}^0,...,J_{N_1}^c)=(J_{N_0+m'}^0,...,J_{N_0+m'}^c)\]
Taking $\rho_1=w[0..N_0]$ and $\rho_2=w[N_0..N_0+m']$ concludes the proof.

The PSPACE upper and lower bound proofs are analogous to the proof of Proposition \ref{p-dec simple almost}.
\end{proof}

\smallskip\noindent{\bf Details of Proposition \ref{p-undec limit buchi}.} 
We prove that the simple limit B\"uchi and coB\"uchi problems are undecidable.
\begin{proof}
 This is a direct consequence of the fact that the simple limit reachability problem is undecidable. 
The reduction from an instance of the simple limit reachability problem is direct: we only delete all outgoing 
transitions from the accepting states in $F$, and transform them into self loops for all label $a\in\Sigma$. 
We get a probabilistic automaton on which the simple limit B\"uchi and coB\"uchi problems are satisfied iff the simple limit reachability problem is satisfied.
\end{proof}

\smallskip\noindent{\bf Details of Proposition \ref{p-LS-ind-simple}.} 
We prove Proposition \ref{p-LS-ind-simple}.
\begin{proof}
 We just have to show that for any $\alpha\in\Delta(Q)$ and $\rho\in\Sigma^*$, the process induced by $\rho^\omega$ and $\alpha$ on $Q$ is simple. 
Let $\lbrace X_n\rbrace_{n\in\mathbb{N}}$ be the non-homogeneous Markov chain induced on $Q$ by $\alpha$ and $\rho^\omega$. Then for all $i\in\set{0,1,\ldots,|\rho|-1}$, 
the chain $\lbrace X_{n\cdot |\rho|+i}\rbrace_{n\in\mathbb{N}}$ is homogeneous. 
The result follows from the classical decomposition Theorem of the state space of an homogeneous Markov chain into periodic components of recursive classes, and transient states.
\end{proof}

\smallskip\noindent{\bf Details of Theorem \ref{thrm:lasso}.} 
We prove Theorem \ref{thrm:lasso}.

\begin{proof}
 By the results of Section \ref{s-decidable pb}, if the \emph{simple} almost or positive parity problem is satisfied, then it is satisfied by a lasso shape word. 
Along with Proposition \ref{p-LS-ind-simple}, this implies that the simple almost (resp. positive) parity problem is equivalent 
to the question whether there is lasso shape word that is accepted with 
probability~1 (resp. with positive probability). Since the simple almost parity problem and the positive parity problem are both PSPACE-complete, the theorem follows.
\end{proof}

\section{Details of Section \ref{s-simple autom}}

\subsection{Details of Sub-Section \ref{ss-SA-struct}}

\begin{comment}
\smallskip\noindent{\bf Details of Lemma \ref{l-cont cycle}.} 
We prove Lemma \ref{l-cont cycle}.
\begin{proof}
Notice first that by definition, every cycle in $\mathcal{H}_\mathcal{A}$ contains a cycle which is elementary. 
Thus, we just have to show that $seq$ contains a cycle. If $A_{k+1}=A_1$, then we are done. 
If $A_{k+1}\subsetneq A_1$, let $B_1$ be such that $A_{k+1}\stackrel{\rho_1}{\rightarrow}\stackrel{\rho_2^\#}{\rightarrow}\ldots\stackrel{\rho_k}{\rightarrow}B_1$. 
Then we have $B_1\subseteq A_{k+1}$. If $B_1=A_{k+1}$ then we are done, since $B_1\stackrel{\rho_1}{\rightarrow}\stackrel{\rho_2^\#}{\rightarrow}\ldots\stackrel{\rho_k}{\rightarrow}B_1$ 
is a subpath of $seq$ which is a cycle. If $B_1\subsetneq A_{k+1}$, then we continue the construction iteratively: for $i\geq1$, until we find $B_{i+1}$ such that $B_i=B_{i+1}$, 
we let $B_{i+1}$ be such that $B_{i}\stackrel{\rho_1}{\rightarrow}\stackrel{\rho_2^\#}{\rightarrow}\ldots\stackrel{\rho_k}{\rightarrow}B_{i+1}$. 
By construction at each step we have $B_i\not=\emptyset$, and $B_{i+1}\subseteq B_i$. Clearly, the construction has to stop after at most $|Q|$ 
steps, and we get a cycle $B_i\stackrel{\rho_1}{\rightarrow}\stackrel{\rho_2^\#}{\rightarrow}\ldots\stackrel{\rho_k}{\rightarrow}B_i$ which is a sub-path of $seq$.
\end{proof}
\end{comment}

\smallskip\noindent{\bf Details of Lemma \ref{l-borne mots rec}.}
We prove the key Lemma \ref{l-borne mots rec}.

\begin{proof}

Given $U\subseteq Q$ and $\rho\in\Sigma^*$, let

\[
\delta^{-1}(\rho)(U)=\lbrace q\in Q\ |\ \delta(q,\rho)(U)>0\rbrace
\]

The following remark will be useful: given $\rho=a_1,...,a_n\in\Sigma^*$, given $U\subseteq Q$ and $i\in\set{0,1,2,\ldots,n-1}$, let $S_i=\delta^{-1}(a_{i+1},...,a_n)(U)$. Then we have:
\begin{compactenum}
 \item For all $i\in\set{0,1,2,\ldots,n-1}$, $\delta(Q\setminus S_i,a_{i+1},...,a_n)\subseteq Q\setminus U$
 \item For all $i\in\set{0,1,2,\ldots,n-2}$, if $\delta(S_i,a_{i+1})\subseteq S_{i+1}$, then $\delta(\alpha,a_1,...,a_{i})(S_i)=\delta(\alpha,a_1,...,a_{i+1})(S_{i+1})$
 \item Given $i\in\set{0,1,2,\ldots,n-2}$, let $k_i$ be the number of integers $j\in[i;n-2]$ such that $\delta(S_j,a_{j+1})\not\subseteq S_{j+1}$. Then, for all $i\in\set{0,1,2,\ldots,n-2}$, 
\[\delta(\alpha,a_1,...,a_n)(U)\geq\delta(\alpha,a_1,...,a_{i})(S_i)\cdot\epsilon^{k_i}\]
\end{compactenum}

The only non trivial point is the last one. It follows from the fact that for all $\rho\in\Sigma^*$ and $q,q'\in Q$, if $\delta(q,\rho)(q')>0$, then by definition of $\epsilon$ we have $\delta(q,\rho)(q')>\epsilon^{|\rho|}$.

By contradiction, suppose that there exists $\rho\in\CRec(q)$ and $U\subseteq Q$ such that $U\subseteq\mathrm{Supp}(\delta(q,\rho))$ and

\[
\delta(q,\rho)(U)<\epsilon^{2^{2\cdot|Q|}} 
\]

We show that then we can write $\rho=\rho_1\cdot \rho_2\cdot \rho_3$ where $\rho_1,\rho_2,\rho_3$ are such that $\delta(q,\rho_1)$ can be partitioned into two subsets $A$ and $B$ such that $(A,B,\rho_2)$ is a 
$\#$-reduction. This contradicts the definition of $\CRec(\mathcal{A})$.

Let $\rho=a_1,...,a_{l}$. Given $i\in\set{0,1,2,\ldots,l-1}$, let:
\begin{compactitem}
 	\item $V^n_i=\delta^{-1}(a_{i+1},a_{i+2},...a_{l})(U)\ \cap\ \delta(q,a_1,\ldots,a_i)$
	\item $W^n_i=(Q\setminus V^n_i)\ \cap\ \delta(q,a_1,\ldots,a_i)$
\end{compactitem}

Using the third point of the previous remark, since $\delta(q,\rho)(U)<\epsilon^{2^{2\cdot|Q|}}$, there exists a least $k$ integers $i$ in $\set{1,2,\ldots,l-2}$ such that $\delta(V^n_i,a_{i+1})\not\subseteq V^n_{i+1}$, where $k$ satisfies $\epsilon^k<\epsilon^{2^{2\cdot|Q|}}$. Thus, $k\geq 2^{2\cdot|Q|}$. Let $n_1,...,n_{2^{2\cdot|Q|}}$ be the $2^{2\cdot|Q|}$ largest integers in $\set{1,2,\ldots,l}$ such that $\delta(V^n_i,a_{i+1})\not\subseteq V^n_{i+1}$.

By a simple cardinality argument, there exist $i<j$ in $\set{1,2,\ldots,2^{2\cdot|Q|}}$ such that $V^n_{n_i}=V^n_{n_j}$ and $W^n_{n_i}=W^n_{n_j}$. Let $\rho_1=a_1,\ldots,a_{n_i-1}$, $\rho_2=a_{n_i},\ldots,a_{n_j-1}$ and $\rho_3=a_{n_j},\ldots,a_n$. Then we are in the following situation:\\
\xymatrix{
&	&	q\ar[r]^{\rho_1}\ar[rd]^{\rho_1}	          &      W^n_{n_i}\ar[r]^{\rho_2}                & W^n_{n_i}\\
&	&	& V^n_{n_i}\ar[r]^{\rho_2}\ar[ur]^{\rho_2} 		   & V^n_{n_i}\\
}
$\ $\\
That is, $\delta(q,\rho_1)$ can be partitionned into two subsets $V^n_{n_i}$ and $W^n_{n_i}$ such that $\delta(W^n_{n_i},\rho_2)\subseteq W^n_{n_i}$, $\delta(V^n_{n_i},\rho_2)\subseteq V^n_{n_i}\cup W^n_{n_i}$, and $\delta(V^n_{n_i},\rho_2)\not\subseteq V^n_{n_i}$. This implies that there exists $A\subseteq V^n_{n_i}$ such that $(A,(V^n_{n_i}\setminus A)\cup W^n_{n_i},\rho_2)$ is a 
$\#$-reduction. Since $A$, $(V^n_{n_i}\setminus A)\cup W^n_{n_i}$ is a partition of $\delta(q,\rho_1)$, we get that the execution tree $(q,\rho)$ contains a 
$\#$-reduction. This is a contradiction since $\rho\in\CRec(q)$.
\end{proof}

\subsection{Extended support graph}\label{s-ext-supp}
In this sub-section, we introduce the necessary preliminary concepts before the formal definition of \emph{structurally simple probabilistic automata}. 
These concepts can be listed as:
\begin{itemize}
 \item A notion of \emph{Linked Graphs}, and how to associate a \emph{Linked Graph} to a word and an automaton. This is a technical way to represent transitions induced on 
a probabilistic automaton by an input word.
 \item A notion of $\#-\rho$ reachability, which generalizes the notion of $\#$-reachability of \cite{gimbert2010probabilistic}.
 \item A notion of \emph{extended support graph} of a probabilistic automaton, which generalizes the notion of \emph{support graph} of \cite{gimbert2010probabilistic}.
\end{itemize}

We start with the notion of \emph{linked graphs}. An example is given after the two following definitions to give more intuition. 
In the following of the sub-section, we call bipartite graph on $Q$ a subset of $Q\times Q$.

\begin{definition}[Linked Graph]
 Given $n\geq1$, a \emph{linked graph} of length $n$ on $Q$ is a sequence $\mathcal{I}=(G_1,\ldots,G_{n})$ of $n$ bipartite graphs: for all 
$i\in 1,\ldots,n$, the bipartite $G_i$ is a set of couples of states of $Q$, i.e. $G_i\subseteq Q\times Q$, such that for all $i\in1,..,n-1$ we have $B_i(\mathcal{I})=A_{i+1}(\mathcal{I})$, where:
\[\forall i\in1,..,n\ \ A_i(\mathcal{I})=\lbrace s\in Q\ |\ \exists t\in Q\ s.t.\ (s,t)\in G_i\rbrace;\ \mathrm{and}\ B_i(\mathcal{I})=\lbrace t\in Q\ |\ \exists s\in Q\ s.t.\ (s,t)\in G_{i}\rbrace\]
The set $A_1=\org(\mathcal{I})$ is called the \emph{origin} of $\mathcal{I}$, and the set 
$B_n=\dest(\mathcal{I})$ is called the \emph{destination} of $\mathcal{I}$.
\end{definition}

Given a linked graph $\mathcal{I}=(G_1,\ldots,G_{n})$, the \emph{compaction} of $\mathcal{I}$, written $\Comp(\mathcal{I})$, is the graph $(V,E)$ where $V\subseteq Q$ and $E\subseteq Q\times Q$, and such that for all $s,t\in Q$ 
we have $(s,t)\in E$ if $s\in A_1, t\in B_n$, and there exists a sequence of 
edges $(s_{1},t_{1})\in G_{1},(s_{2},t_{2})\in G_{2},\ldots,(s_{n},t_{n})\in G_{n}$ such that $s_{1}=s$ and $t_{n}=t$. 

Given an edge-oriented graph $\mathcal{G}=(V,E)$, a \emph{terminal component} is a set of states $A\subseteq V$ such that given any $s,t\in A$, there exists a path between $s$ and $t$ in $A$.

If $\dest(\mathcal{I})\subseteq \org(\mathcal{I})$, we define $\Rec(\mathcal{I})$ as the set of $t\in V$ which belong to a terminal component of  $\Comp(\mathcal{I})$.
Given $s\in \org(\mathcal{I})$, we define $\Rec(s,\mathcal{I})$ as the set of $t\in \Rec(\mathcal{I})$ which are reachable from $s$ in $\Comp(\mathcal{I})$.

\begin{definition}
 Let $n\geq1$, let $\mathcal{A}$ be a probabilistic automaton, and let $\rho=a_1,...,a_n\in\Sigma^n$. Given $A\subseteq Q$, 
we define inductively the linked graph $\LG(\rho,A,\mathcal{A})=(G_1,\ldots,G_n)$ 
induced by $\rho$ and $A$ on $\mathcal{A}$ as follows:
\begin{itemize}
 \item $G_1=\lbrace (s,t)\ |\ s\in A\ \mathrm{and}\ t\in Supp(\delta(s,a_1)))\rbrace$.
 \item For all $i\in1,\ldots n-1$, $G_{i+1}=\lbrace (s,t)\ |\ s\in \dest(G_1,...,G_i)\ \mathrm{and}\ t\in Supp(\delta(s,a_{i+1}))\rbrace$. 
\end{itemize}
\end{definition}

\begin{example}\label{e-graph-LG}
Consider the following automaton:

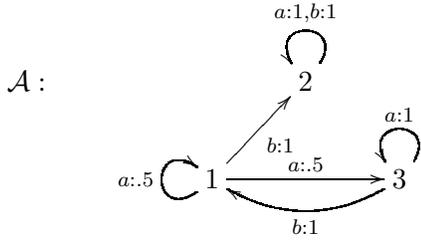
\xymatrix{
\mathcal{A}: &&		&	2\ar@(ur,ul)[]_{a:1,b:1}\\
&&1\ar@(dl,ul)[]^{a:.5}\ar[rr]^{a:.5}\ar[ru]_{b:1}	&&	3\ar@/^1pc/[ll]^{b:1}\ar@(ur,ul)[]_{a:1}
}

Let $\rho=a\cdot b\cdot a$. The state space of $\mathcal{A}$ is $Q=\lbrace1,2,3\rbrace$, and the linked graph $\mathcal{I}=\LG(\rho,\lbrace 1,2,3\rbrace,\mathcal{A})$ can be represented as:

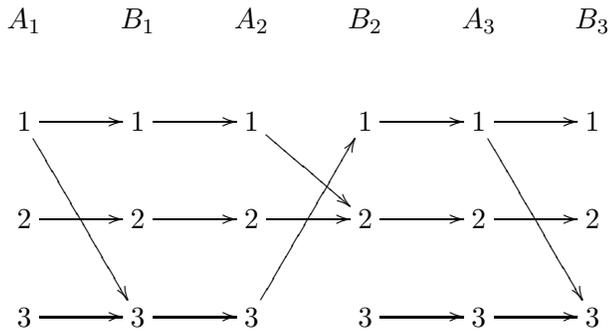
\xymatrix{
A_1		&	B_1		&	A_2		&	B_2			&	A_3			&	B_3\\
1\ar[ddr]\ar[r]	&	1\ar[r]		&	1\ar[rd]	&	1\ar[r]			&	1\ar[rdd]\ar[r]		&	1\\
2\ar[r]		&	2\ar[r]		&	2\ar[r]		&	2\ar[r]			&	2\ar[r]			&	2\\
3\ar[r]		&	3\ar[r]		&	3\ar[uur]	&	3\ar[r]			&	3\ar[r]			&	3\\
\\
}

Using a compact representation which avoids repetitions between $B_i$ and $A_{i+1}$, we describe the previous linked graph as:

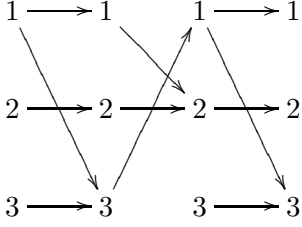
\xymatrix{
1\ar[ddr]\ar[r]	&	1\ar[rd]	&	1\ar[rdd]\ar[r]		&	1\\
2\ar[r]		&	2\ar[r]		&	2\ar[r]			&	2\\
3\ar[r]		&	3\ar[uur]	&	3\ar[r]			&	3\\
\\
}

We have $\org(\mathcal{I})=\dest(\mathcal{I})=\lbrace1,2,3\rbrace$, and $\Comp(\mathcal{I})$ is the graph:

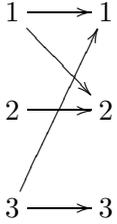
\xymatrix{
1\ar[r]\ar[dr]	&	1\\
2\ar[r]		&	2\\
3\ar[r]\ar[uur]	&	3\\
\\
}

In this case, we have $\Rec(\mathcal{I})=\lbrace 2\rbrace$, and for all $s\in\lbrace1,2,3\rbrace$ we have $\Rec(s,\mathcal{I})=\lbrace2\rbrace$.
\end{example}

\begin{definition}[Borders]
 A \emph{border} of a linked graph $\mathcal{I}$ of length $n$ is a couple of integers $b=(n_1,n_2)$ where $1\leq n_1<n_2\leq n$ and $B_{n_2}(\mathcal{I})\subseteq A_{n_1}(\mathcal{I})$.
\end{definition}



\begin{definition}[Action of a border on a linked graph]
Let $\mathcal{I}=(G_1,\ldots,G_n)$ be a linked graph on $Q$, and let $b=(n_1,n_2)$ be a border on $\mathcal{I}$. 
We define $\LG(b,\mathcal{I})=(G_1',\ldots,G_n')$, the linked graph induced by the action of $b$ on $\mathcal{I}$, as follows:
\begin{itemize}
 \item For all $i<n_1-1$ we let $G_i'=G_i$
 \item We define $G_{n_1-1}',G_{n_1}',...,G_n'$ inductively. First, let $\mathcal{J}=G_{n_1},...,G_{n_2}$. By hypothesis, we have $\dest(\mathcal{J})\subseteq \org(\mathcal{J})$. 
  \begin{itemize}
    \item Given $s,t\in Q$, let $(s,t)\in G'_{n_1-1}$ if $s\in \dest(G_1',...,G_{n_1-2}')$ and $t\in \Rec(s,\mathcal{J})$. 

    \item Suppose that we have defined $G_1',...,G_i'$ until $i\geq n_1-1$. Then we define 
$G_{i+1}'=\lbrace (s,t)\ |\ s\in  \dest(G_1',...,G_{i}')\ \mathrm{and}\ (s,t)\in G_{i+1}\rbrace$.
    
  \end{itemize}
\end{itemize}

\end{definition}

Intuitively, the action of a border $b=(n_1,n_2)$ on a Linked Graph corresponds to a $\#$-transition: we keep only the states of $G_{n_1}$ which are recurrent for 
the sub-Linked Graph placed between $n_1$ and $n_2$, and their successors.

\begin{example}
 Consider the same automaton as in example \ref{e-graph-LG}. As before, let $\rho=a\cdot b\cdot a$ and let $\mathcal{I}=\LG(\rho,\lbrace1,2,3\rbrace,\mathcal{A})$. 
Then $b=(1,2)$ is a border on $\mathcal{I}$ since $B_2(\mathcal{I})=\lbrace1,2,3\rbrace=A_1(\mathcal{I})$. Then the graph $\LG(b,\mathcal{I})$ is the following graph:

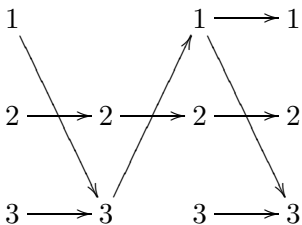
\xymatrix{
1\ar[ddr]	&			&	1\ar[rdd]\ar[r]		&	1\\
2\ar[r]		&	2\ar[r]		&	2\ar[r]			&	2\\
3\ar[r]		&	3\ar[uur]	&	3\ar[r]			&	3
}

\end{example}

\begin{definition}[Chain of borders]
Let $\mathcal{I}=(G_1,\ldots,G_n)$ be a linked graph on $Q$. We call \emph{chain of borders} a sequence 
$B=((n_1^1,n_2^1),(n_1^2,n_2^2),\ldots,(n_1^k,n_2^k))$ of borders such that 
for all $i\in1,\ldots k-1$ we have that $b_{i+1}$ is a border of $\LG(b_{i-1},\LG(b_{i-2},\ldots,\LG(b_1,\mathcal{I})\ldots))$. We define $\LG(B,\mathcal{I})$, the action of the sequence of borders  $B$ 
on $\mathcal{I}$, as the linked graph $\LG(B,\mathcal{I})=\LG(b_{k},\LG(b_{k-1},\ldots,\LG(b_1,\mathcal{I})\ldots))$.
\end{definition}

Intuitively, the action of a chain of borders on a Linked Graph corresponds to a finite sequence of imbricated $\#$-transitions.




In the following definition we consider graphs whose nodes are subsets of $Q$, and whose edges are labeled by bipartite graphs $I\subseteq Q\times Q$. Given $C,D\subseteq Q$ and $I\subseteq Q\times Q$, 
we write $(C,I,D)$ for an edge labeled by the graph $I$ between the 
vertice $C$ and $D$. Given $I\subseteq Q\times Q$, we define:
\[\leftc(I)=\lbrace s\in Q\ |\ \exists t\in Q\ s.t.\ (s,t)\in I\rbrace\ \mathrm{and}\ \rightc(I)=\lbrace t\in Q\ |\ \exists s\in Q\ s.t.\ (s,t)\in I\rbrace\]

\begin{definition}[Complete linked graphs]
 A \emph{complete linked graph} on $Q$ is a graph $\mathcal{G}=(V,E)$ where $V=\mathcal{P}(Q)$ and the edges of $\mathcal{G}$ are labeled by bipartite graphs on $Q$ such that for all edge $(A,I,B)\in E$, 
we have $A=\leftc(I)$ and $B=\rightc(I)$.
\end{definition}

A path $p=(A_1,I_1,A_2,...,I_{n-1},A_{n+1})$ on a complete linked graph naturally induces a linked graph $(I_1,...,I_{n})$.

\begin{definition}[$\#$-reachability]
Let $C,D\subseteq Q$. Then:
\begin{itemize}
 \item Given $\rho\in\Sigma^*$, we say that $D$ is $\#-\rho$-reachable from $C$, written $C\stackrel{\#-\rho}{\rightarrow}D$, if there 
exists a chain of borders $B$ on $\mathcal{I}=\LG(C,\rho,\mathcal{A})$ such that $D=\dest(\LG(B,\mathcal{I}))$. 
 \item We say that $D$ is $\#$-reachable from $C$, written $C\stackrel{\#}{\rightarrow}D$, if there exists $\rho\in\Sigma^*$ such that $C\stackrel{\#-\rho}{\rightarrow}D$. 

\end{itemize}

Let $C,D\subseteq Q$ and $I\subseteq Q$. Then:
\begin{itemize}
 \item Given $\rho\in\Sigma^*$, we say that $D$ is $\#-\rho-I$-reachable from $C$, written $C\stackrel{\#-\rho-I}{\rightarrow}D$, if there 
exists a chain of borders $B$ on $\mathcal{I}=\LG(C,\rho,\mathcal{A})$ and $i\in1,...,|\rho|$ such that given $\LG(B,\mathcal{I})=(G_1,...,G_n)$, 
we have $\Comp(G_1,...,G_i)=I$ and $\dest(G_1,...,G_i)=D$.
 \item We say that $D$ is $\#-I$-reachable from $C$, written $C\stackrel{\#-I}{\rightarrow}D$, if there exists $\rho\in\Sigma^*$ such that $C\stackrel{\#-\rho-I}{\rightarrow}D$. 
\end{itemize}

\end{definition}

\begin{example}[$\#$-reachability]\label{e-simple-autom}$\ $
\begin{multicols}{2}
    \begin{enumerate}[]
      \item 

Consider the following probabilistic automaton:\\
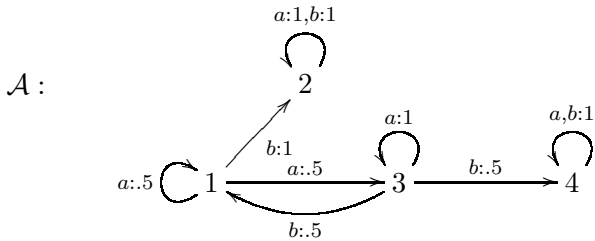
\xymatrix{
\mathcal{A}: &&		&	2\ar@(ur,ul)[]_{a:1,b:1}\\
&&1\ar@(dl,ul)[]^{a:.5}\ar[rr]^{a:.5}\ar[ru]_{b:1}	&&	3\ar[rr]^{b:.5}\ar@/^1pc/[ll]^{b:.5}\ar@(ur,ul)[]_{a:1}	&& 4\ar@(ur,ul)[]_{a,b:1}
}

      \item 
Let $\rho=a\cdot a\cdot a\cdot b\cdot a\cdot a\cdot b$. Let $\mathcal{I}=\LG(1,\rho)$, let $b_1=(1,2)$, $b_2=(5,6)$, and $b_3=(4,7)$. 
Then $B=(b_1,b_2,b_3)$ is a chain of borders for $\mathcal{I}$, and moreover we have $\dest(\LG(B,\mathcal{I}))=\lbrace 4\rbrace$. 
Thus $\lbrace4\rbrace$ is $\#$-reachable from $\lbrace1\rbrace$. Notice however that $\lbrace 4\rbrace$ is not 
reachable from $\lbrace1\rbrace$ in the support graph of $\mathcal{A}$.

Indeed, we have 
\begin{center}
$\mathcal{I}=(\lbrace (1,3),(1,1)\rbrace,
\lbrace (1,3),(1,1),(3,3)\rbrace,$\\
$\lbrace (1,2),(3,1),(3,4)\rbrace
\lbrace (2,2),(1,3),(1,1),(4,4)\rbrace,$\\
$\lbrace (1,1),(1,3),(3,3),(2,2),(4,4)\rbrace,$\\
$\lbrace (1,2),(2,2),(4,4),(3,1),(3,4)\rbrace)$
\end{center}

Now, if $b_1=(1,2)$, $b_2=(5,6)$, and $b_3=(4,7)$. Then $B=(b_1,b_2,b_3)$ is a chain of borders for $\mathcal{I}$, and moreover we have:
\begin{center}
$\LG(b_1,\mathcal{I})=(\lbrace (1,3)\rbrace,
\lbrace (3,3)\rbrace,$\\
$\lbrace (3,1),(3,4)\rbrace
\lbrace (1,3),(1,1),(4,4)\rbrace,$\\
$\lbrace (1,1),(1,3),(3,3),(4,4)\rbrace,$\\
$\lbrace (1,2),(4,4),(3,1),(3,4)\rbrace)$
\end{center}
Next
\begin{center}
$\LG((b_1,b_2),\mathcal{I})=(\lbrace (1,3)\rbrace,
\lbrace (3,3)\rbrace,$\\
$\lbrace (3,1),(3,4)\rbrace
\lbrace (1,3),(1,1),(4,4)\rbrace,$\\
$\lbrace (1,3),(3,3),(4,4)\rbrace,$\\
$\lbrace (4,4),(3,1),(3,4)\rbrace)$
\end{center}
And finally 
\begin{center}
$\LG((b_1,b_2,b_3),\mathcal{I})=(\lbrace (1,3)\rbrace,
\lbrace (3,3)\rbrace,$\\
$\lbrace (3,1),(3,4)\rbrace
\lbrace (4,4)\rbrace,$\\
$\lbrace (4,4)\rbrace,$\\
$\lbrace (3,4)\rbrace)$
\end{center}

Thus, $\dest(\LG(B,\mathcal{I}))=\lbrace 4\rbrace$, and $\lbrace4\rbrace$ is $\#$-reachable from $\lbrace1\rbrace$. Notice however that $\lbrace 4\rbrace$ is not 
reachable from $\lbrace1\rbrace$ in the support graph of $\mathcal{A}$.


    \end{enumerate}
  \end{multicols}
\end{example}

\begin{definition}[Extented support graph]
Let $\mathcal{A}$ be a probabilistic automaton. We define $\mathcal{H}_\mathcal{A}=(\mathcal{P}(Q),E)$, the \emph{extended support graph of $\mathcal{A}$}, such that 
given $C,D\subseteq Q$ and $I\subseteq Q\times Q$, we have $(C,I,D)\in E$ if $C\stackrel{\#-I}{\rightarrow}D$.
\end{definition}

Clearly $\mathcal{H}_\mathcal{A}$ is a complete linked graph.

\begin{proposition}\label{t-constr-ext-supp-graph}
 We can construct $\mathcal{H}_\mathcal{A}$ in EXPSPACE.
\end{proposition}

We present the following Algorithm to compute a graph $G_N$ in EXPSPACE, 
and the following Proposition shows that $G_N=\mathcal{H}_\mathcal{A}$. First, we need a preliminary Lemma, which can be proved using a simple counting argument.

\begin{lemma}\label{l-reduction}
 Let $\mathcal{G}$ be a complete linked graph. Let $A\subseteq Q$, let $I\subseteq Q$, and let $p=(A_1,I_1,...,I_{n},A_{n+1})$ be a path on $\mathcal{G}$ such that $A_1=A$. 
Suppose that there exists a border $b$ of $(I_1,...,I_n)$ 
and $i\in1,...,n$ such that given $\LG(b,\mathcal{I})=(G_1,...,G_n)$, 
we have $\Comp(G_1,...,G_i)=I$. Then there exists a path $p=(A_1',I_1',...,I_{m}',A_{m+1}')$ on $\mathcal{G}$ of length at most $2^{|Q|}$, where $A_1'=A$, and a border $b$ of $(I_1',...,I_m')$ 
and $i\in1,...,m$ such that given $\LG(b,\mathcal{I})=(G_1',...,G_m')$, we have $\Comp(G_1',...,G_i')=I$.
\end{lemma}

\begin{algorithm}[Computation of the extended support graph]\label{a-comp_SG}
Let $\mathcal{A}$ be a probabilistic automaton.

\begin{itemize}
  \item Let $G_0=(\mathcal{P}(Q),E_0)$ be the transitions labeled graph such that given $A,B\subseteq Q$ and $I\subseteq Q\times Q$ we have $(A,I,B)\in E_0$ if there 
exists $\rho\in\Sigma^*$ such that $\delta(A,\rho)=B$ and $\Comp(\LG(A,\rho,\mathcal{A}))=I$.
  \item Given $i\geq1$, let $G_i=(\mathcal{P}(Q),E_i)$ be the graph with transitions labeled by bipartite graphs such that given $A,B\subseteq Q$ and $I\subseteq Q\times Q$ we have $(A,I,B)\in E_i$ either 
if $(A,I,B)\in E_{i-1}$, or if there exists a path $\mathcal{J}=(H_1,H_2,...,H_n)$ in $G_{i-1}$ (hence also a linked graph), and a border $b$ of $\mathcal{J}$, 
	  such that $A=\org(\mathcal{J})$, and if $\LG(b,\mathcal{J})=(G_1,...,G_n)$ there exists $i\in1,...,n$ such that $B=\dest(G_1,...,G_i)$, and $\Comp(G_1,...,G_i)=I$. Using Lemma \ref{l-reduction}, 
we can decide this fact in EXPSPACE.
\end{itemize}
The graphs $G_0,...,$ all have vertice in $\mathcal{P}(S)$, and their edges are labeled by bipartite graphs. Moreover, the sequence of the $G_i$ is increasing (for the inclusion relation). 
The maximum number of different bipartite graphs between the 
subsets of $Q$ is bounden by $|Q|^{|Q|}\cdot2^{|Q|}$. As a consequence, the sequence 
of the $G_i$ must stabilize after some $N\leq |Q|^{|Q|}\cdot2^{2\cdot|Q|}$, and we can build $G_N$ in EXPSPACE.
\end{algorithm}

\begin{proposition}\label{l-algo-expaspace}
 Let $G_0,...,G_N$ be given by the Algorithm \ref{a-comp_SG}. Then we have $G_N=\mathcal{H}_\mathcal{A}$
\end{proposition}
\begin{proof}
 First, we prove by induction that for all $i\in0,..,N$, all $A,B\subseteq Q$ and $I\subseteq Q\times Q$, if $(A,I,B)$ is an edge of $G_i$, then $(A,I,B)$ is also an edge of $\mathcal{H}_\mathcal{A}$.
The case $i=0$ is trivial. Let $i\geq1$, let $A,B\subseteq Q$ and $I\subseteq Q\times Q$, and suppose that $(A,I,B)$ is an edge of $G_i$. If $(A,I,B)$ is an edge of $G_{i-1}$ we are 
done by induction hypothesis. If not, then there exists a path $\mathcal{J}=(H_1,H_2,...,H_n)$ in $G_{i-1}$ and a border $b$ of $\mathcal{J}$ 
	  such that $A=\org(\mathcal{J})$, and if $\LG(b,\mathcal{J})=(G_1,...,G_n)$ there exists $i\in1,...,n$ such that $B=\dest(G_1,...,G_i)$, and $\Comp(G_1,...,G_i)=I$. Now, by induction hypothesis, 
$\mathcal{J}$ is also a path in $\mathcal{H}_\mathcal{A}$, and by definition of $\mathcal{H}_\mathcal{A}$, the linked graph $\LG(b,\mathcal{J})$ is also in $\mathcal{H}_\mathcal{A}$, which implies 
that $(A,I,B)$ is an edge of $\mathcal{H}_\mathcal{A}$.

Conversely, suppose that $(C,I,D)$ is an edge of $\mathcal{H}_\mathcal{A}$. Let $\rho\in\Sigma^*$, $i\in1,...,|\rho|$ and let $B$ be the chain of borders on $\mathcal{I}=\LG(D,\rho,\mathcal{A})$ such that 
$A_i(\LG(B,\mathcal{I}))=D$ and if $\LG(B,\mathcal{I})=(G_1,...,G_{|\rho|})$ we have $\Comp(G_1,...,G_i)=I$. Then we can execute the border action on the set of graphs $G_0,...,G_k$ where $k$ is 
the number of borders in $B$, to get the edge $(C,I,D)$ in $G_k$. This proves the result.
\end{proof}

\subsection{Details of Sub-Section \ref{ss-SA-struct} - Continued}

\smallskip\noindent{\bf Details of Lemma \ref{l-ult pos impl simple}.} 
We prove Lemma \ref{l-ult pos impl simple}.

\begin{proof}
For all $n\in\mathbb{N}$, we let $A_n=\mathrm{Supp}(\delta(\alpha,w[1..n]))$ and $B_n=Q\setminus A_n$. 
By hypothesis, for all all $n\geq N$ and all $q\in A_n$, we have $\mu_n^w(q)=\delta(\alpha,w[1..n])(q)>\gamma$. 
Moreover, for all $n$ and all $q\in B_n$ we have $\mu^w_n(q)=0$. This shows that the process is simple.
\end{proof}

\smallskip\noindent{\bf Details of Lemma \ref{l-exec trees impl simple}.} 
We prove Lemma \ref{l-exec trees impl simple}.

\begin{proof}
Let $\rho\in\Sigma^\omega$, and let $\lbrace\mu^\rho_n\rbrace_{n\in\mathbb{N}}$ be the process induced on $Q$ by $\rho$. 
By hypothesis, let $(\alpha,\rho_1),\rho_1',(\alpha_2,\rho_2),\rho_2',...(\alpha_k,\rho_k)$ be a sub-sequence of recurrent execution 
trees of $(\alpha,\rho)$ of length at most $K$. That is, we have $\sum_{i=1}^{k-1}|\rho_i'|\leq K$. 
By definition, for all $i\in\set{1,\ldots k-1}$ we have $\rho_i\in\Sigma^*$ and $\rho_i'\in\Sigma^*$, and $\rho_k\in\Sigma^\omega$.
For all $i\in\set{1,\ldots,k-1}$, let $\alpha_i'=\delta(\alpha_i,\rho_i)$. We are in the following situation:
\[\alpha\stackrel{\rho_1}{\rightarrow}\alpha_1'\stackrel{\rho_1'}{\rightarrow}\alpha_2\stackrel{\rho_2}{\rightarrow}\alpha_2'\stackrel{\rho_2'}{\rightarrow}\alpha_3\ldots\stackrel{\rho_{k-1}'}{\rightarrow}\alpha_{k}\stackrel{\rho_k}{\rightarrow}\]
We know that: 
\begin{compactitem}
 \item For all $i\in\set{1,\ldots,k-2}$, the execution tree $(\alpha_i,\rho_i)$ is chain recurrent
 \item $(\alpha_{k},\rho_k)$ is chain recurrent

\end{compactitem}

We show that the process $\lbrace\mu^\rho_n\rbrace_{n\in\mathbb{N}}$ satisfies the hypothesis of Lemma \ref{l-ult pos impl simple}. 
As before, let $\epsilon=\epsilon(\mathcal{A})$ be the minimal non zero probability which appears among the values $\delta(q,a)(q')$ 
when $q,q'\in Q$ and $a\in\Sigma$. Let $\lambda=\epsilon^{2^{2\cdot|Q|}}$. By Lemma \ref{l-borne mots rec}, for all $q\in Q$, all 
$\rho'\in\CRec(q)$ and all $q'\in\mathrm{Supp}(\delta(q,\rho'))$, we have $\delta(q,\rho')(q')\geq\lambda$. We claim that for all 
$i\in\set{1,\ldots,k-1}$ and all $q\in\mathrm{Supp}(\alpha_i')$, we have $\alpha_i'(q)\geq(\mathrm{Min}_{q\in\mathrm{Supp}(\alpha)}\alpha(q))\cdot\lambda^i\cdot\epsilon^{K\cdot i}$. 
We prove this result by induction on $i$:
\begin{compactitem}
 \item The case $i=1$ follows from the use of Lemma \ref{l-borne mots rec} on the chain recurrent execution tree $(\alpha_1,\rho_1)$.
 \item Suppose the proposition true until $i\in\set{1,\ldots,k-2}$. Let $q'\in\mathrm{Supp}(\alpha_{i+1}')$. then there exists $q\in\mathrm{Supp}(\alpha_{i}')$ such that $\delta(q,\rho_i\cdot\rho_i')(q')>0$. Let $q''\in Q$ be such that $\delta(q,\rho_i)(q'')>0$, and $\delta(q'',\rho_i')(q')>0$. By the use of Lemma \ref{l-borne mots rec} on the chain recurrent execution tree $(\alpha_i,\rho_i)$, we know that $\delta(q,\rho_i)(q'')>\lambda$. By definition of $\epsilon$ and $K$, we have that $\delta(q'',\rho_i')(q')\geq\epsilon^{|\rho_i'|}$, hence $\delta(q'',\rho_i')(q')\geq\epsilon^{K}$. We have $\alpha_{i+1}'(q')\geq\alpha_i(q)\cdot\delta(q,\rho_i\cdot\rho_i')(q')$. Since by induction hypothesis we have that $\alpha_i(q)\geq(\mathrm{Min}_{q\in\mathrm{Supp}(\alpha)}\alpha(q))\cdot\lambda^i\cdot\epsilon^{K\cdot i}$, we get that $\alpha_{i+1}'(q')\geq(\mathrm{Min}_{q\in\mathrm{Supp}(\alpha)}\alpha(q))\cdot\lambda^{i+1}\cdot\epsilon^{K\cdot {i+1}}$, hence the result.
\end{compactitem}

Now, let $N=\sum_{i=1}^{k-1}(|\rho_i|+|\rho_i'|)$, and let $n\geq N$. Since $(\alpha_{k},\rho_k)$ is chain recurrent, 
we can apply the same method for the chain recurrent execution tree $(\alpha_k,\rho_k)$. 
As a conclusion, we see that the process $\lbrace\mu^\rho_n\rbrace_{n\in\mathbb{N}}$ satisfies the hypothesis of Lemma \ref{l-ult pos impl simple} 
with the parameters $N=\sum_{i=1}^{k-1}(|\rho_i|+|\rho_i'|)$ and $\gamma=(\mathrm{Min}_{q\in\mathrm{Supp}(\alpha)}\alpha(q))\cdot\lambda^{K}\cdot\epsilon^{K\cdot {K}}$. 
This proves the result.
\end{proof}

\begin{comment}
\smallskip\noindent{\bf Details of Lemma \ref{l-ssimple sharp reach chain rec}.} 
We prove Lemma \ref{l-ssimple sharp reach chain rec}.

\begin{proof}
Let $\rho=\rho_1\cdot\rho_2\cdot\ldots\cdot\rho_{2\cdot k-1}$ be the decomposition of $\rho$ into sub-words such that
\[A\stackrel{\rho_1}{\longrightarrow}\stackrel{\rho_2^\#}{\longrightarrow}\stackrel{\rho_3}{\longrightarrow}\stackrel{\rho_4^\#}{\longrightarrow}\ldots\stackrel{\rho_{2\cdot k-1}}{\longrightarrow}B\]
By Lemma \ref{l-cont cycle}, this path contains an elementary cycle:
\[C_1\stackrel{\rho_1}{\longrightarrow}C_2\stackrel{\rho_2^\#}{\longrightarrow}C_3\stackrel{\rho_3}{\longrightarrow}C_4\stackrel{\rho_4^\#}{\longrightarrow}\ldots\stackrel{\rho_{2\cdot k}^\#}{\longrightarrow}C_1\]
Since $\mathcal{A}$ is structurally simple, the cycle does not contain any $\#$-reduction, hence for all $i\in\set{1,\ldots,2\cdot k}$ which is even, we have $C_i=C_i\cdot\rho_i^\#$, and then also $C_i=C_i\cdot\rho_i$. By Lemma \ref{l-elem cycl chain rec}, this implies that $(C_1,\rho)$ is chain recurrent.
\end{proof}
\end{comment}

\smallskip\noindent{\bf Details of Lemma \ref{l-s simple impl exec trees}.} 
We prove Lemma \ref{l-s simple impl exec trees}.

\begin{proof}
We build iteratively the following sequences $\set{A_i}_{i\in\mathbb{N}}$, $\set{A_i'}_{i\in\mathbb{N}}$, $\set{B_i}_{i\in\mathbb{N}}$, $\set{\rho_i}_{i\in\mathbb{N}}$, 
$\set{\rho_i'}_{i\in\mathbb{N}}$, $\set{a_i}_{i\in\mathbb{N}}$, $\set{w_i}_{i\in\mathbb{N}}$:
\begin{compactitem}
 \item Let $A_1=\mathrm{Supp}(\alpha)$, and $B_1=\emptyset$.
 \item $\rho_1$ is the longest prefix of $w$ such that there exists $A_1''\subseteq Q$ such that 
$A_1\stackrel{\#\text{-}\rho_1}{\longrightarrow}A_1''$ and $A_1''\subseteq A_1$, and $A_1'$ is minimal among the set of $A_1''\subseteq Q$ such that 
$A_1\stackrel{\#\text{-}\rho_1}{\longrightarrow}A_1''$ and $A_1''\subseteq A_1$ (minimal for the inclusion. If there exists several possibilities for $A_1'$ we just pick one of them). 
If the set of valid words is not bounded, then we let $\rho_1=w$ and the construction stops. If there exists no $\rho_1$ prefix of $w$ such that there exists $A_1''\subseteq Q$ such that 
$A_1\stackrel{\#\text{-}\rho_1}{\longrightarrow}A_1''$ and $A_1''\subseteq A_1$, then we let $\rho_1=\epsilon$ be the empty word and we continue the construction.
 \item Let $w_1$ be such that $\rho_1\cdot w_1=w$. 
 \item If the execution tree $(A_1',w_1)$ is chain recurrent, we stop the construction. If not, let $\rho_1'$ of maximal length be such that $(A_1',\rho_1')$ is chain recurrent.
 \item Let $a_1\in\Sigma$ be the letter which follows $\rho_1\cdot\rho_1'$ in $w$. By construction, $(A_1',\rho_1'\cdot a_1)$ is not chain recurrent, hence we can decompose $\rho_1'\cdot a_1$ as $\rho_1'\cdot a_1=\rho_1''\cdot\rho_1'''$ and find $U_1,V_1$ a partition of $\delta(A_1',\rho_1'')$ such that $(U_1,V_1,\rho_1''')$ is a
 $\#$-reduction and $U_1\cup V_1=\delta(A_1',\rho_1'')=\delta(A_1',\rho_1'\cdot a_1)$. We let $A_2=V_1$. Remark that $A_1'\stackrel{\#\text{-}\rho_1'\cdot a_1}{\longrightarrow}A_2$. Since $A_1\stackrel{\#\text{-}\rho_1}{\longrightarrow}A_1'$, this implies that $A_1\stackrel{\#\text{-}\rho_1\cdot\rho_1'\cdot a_1}{\longrightarrow}A_2$. By definition of $A_1'$, we have $A_2\not=A_1$.
 \item Let $B_2=\delta(A_1,\rho_1\cdot\rho_1'\cdot a_1)\setminus A_2$.

 \item Let $i\geq1$. Suppose that we have constructed the sets $A_1,B_1,A_1',\ldots A_{i+1}, B_{i+1}$, and the sequence of finite words $\rho_1,w_1,\rho_1',a_1,\rho_1'',\ldots\rho_i,w_i,\rho_i',a_i,\rho_i''$. We continue the construction as follows:

\begin{compactitem}
	\item $\rho_{i+1}$ is the longest prefix of $w_i$ such that there exists $A_{i+1}''\subseteq Q$ such that 
$A_{i+1}\stackrel{\#\text{-}\rho_{i+1}}{\longrightarrow}A_{i+1}''$ and $A_{i+1}''\subseteq A_{i+1}$, and $A_{i+1}'$ is minimal among the set of $A_{i+1}''\subseteq Q$ such that 
$A_{i+1}\stackrel{\#\text{-}\rho_{i+1}}{\longrightarrow}A_{i+1}''$ and $A_{i+1}''\subseteq A_{i+1}$ (minimal for the inclusion. If there exists several possibilities for $A_{i+1}'$ we just pick one of them). 
If the set of available words is not bounded, then we let $\rho_{i+1}=w_i$ and the construction stops. If there exists no $\rho_{i+1}$ prefix of $w_i$ such that there exists $A_{i+1}''\subseteq Q$ such that 
$A_{i+1}\stackrel{\#\text{-}\rho_1}{\longrightarrow}A_{i+1}''$ and $A_{i+1}''\subseteq A_{i+1}$, then we let $\rho_{i+1}=\epsilon$ be the empty word and we continue the construction.
	\item Let $w_{i+1}$ be such that $\rho_1\cdot\rho_1'\cdot\rho_1''\cdot...\cdot\rho_{i+1}\cdot w_{i+1}=w$.
        \item If the execution tree $(A_{i+1}',w_{i+1})$ is chain recurrent, we stop the construction. If not, let $\rho_{i+1}'$ of maximal length be such that $(A_{i+1}',\rho_{i+1}')$ is chain recurrent.
	 \item Let $a_{i+1}\in\Sigma$ be the letter which follows $\rho_1\cdot\rho_1'\cdot\ldots\cdot\rho_{i+1}\cdot\rho_{i+1}'$ in $w$. By construction, $(A_{i+1}',\rho_{i+1}'\cdot a_{i+1})$ is not chain recurrent, hence we can decompose $\rho_{i+1}'\cdot a_{i+1}$ as $\rho_{i+1}'\cdot a_{i+1}=\rho_{i+1}''\cdot\rho_{i+1}'''\cdot\rho_{i+1}''''$ and find $U_{i+1},V_{i+1}$ a partition of $\delta(A_{i+1}',\rho_{i+1}'')$ such that $(U_{i+1},V_{i+1},\rho_{i+1}''')$ 
is a $\#$-reduction and $U_{i+1}\cup V_{i+1}=\delta(A_{i+1}',\rho_{i+1}'')=\delta(A_{i+1}',\rho_{i+1}'\cdot a_{i+1})$. We let $A_{i+2}=V_{i+1}$. Remark that $A_{i+1}'\stackrel{\#\text{-}\rho_{i+1}'\cdot a_{i+1}}{\longrightarrow}A_{i+2}$. Since by induction we have $A_1\stackrel{\#\text{-}\rho_1}{\longrightarrow}A_{i+1}$ and since by hypothesis $A_{i+1}\stackrel{\#\text{-}\rho_{i+1}}{\longrightarrow}A_{i+1}'$ , this implies that $A_1\stackrel{\#\text{-}\rho_1\cdot\rho_1'\cdot a_1\cdot\ldots\cdot\rho_{i+1}'\cdot a_{i+1}}{\longrightarrow}A_{i+2}$. By definition of $A_{i+1}'$, we have $A_{i+2}\not=A_{i+1}$, and by induction we have $A_{i+2}\not= A_j$ for all $j\leq i+1$.
	\item Let $B_{i+2}=\delta(A_1,\rho_1\cdot\rho_1'\cdot a_1\cdot\ldots\cdot\rho_{i+1}\cdot\rho_{i+1}'\cdot a_{i+1})\setminus A_{i+1}$.

\end{compactitem}

\end{compactitem}

Since there exists at most $2^{|Q|}$ different subsets $A_i$ of $Q$, the construction stops after at most $2^{|Q|}$ steps. We get a sequence:
\[A_1 \stackrel{\rho_1}{\longrightarrow} \stackrel{\rho_1'}{\longrightarrow} \stackrel{a_1}{\rightarrow} \ldots \stackrel{\rho_i}{\longrightarrow}
    \stackrel{\rho_i'}{\longrightarrow} \stackrel{a_i}{\rightarrow} A_{i+1}\stackrel{\rho_{i+1}}{\longrightarrow} \]
Where $\rho_{i+1}\in\Sigma^\omega$. Moreover, we now by construction that $\rho_{i+1}=w_i$, since by hypothesis the set of prefixes $\rho_{i+1}$ of $w_i$ such that there exists $A_{i+1}'\subseteq Q$ such that 
$A_{i+1}\stackrel{\#\text{-}\rho_{i+1}}{\longrightarrow}A_{i+1}'$ and $A_{i+1}'\subseteq A_{i+1}$ is not bounded. We can use the fact that $\mathcal{A}$ is structurally simple iteratively to show that this imply that there 
exists $C\subseteq A_{i+1}$ such that $(C,\rho_{i+1})$ is chain recurrent. Indeed, since $\mathcal{A}$ is structurally simple, to any finite length prefix $pref$ of 
$w_i$ such that there exists $A_{i+1}'\subseteq Q$ such that 
$A_{i+1}\stackrel{\#\text{-}pref}{\longrightarrow}A_{i+1}'$ and $A_{i+1}'\subseteq A_{i+1}$, we can associate $C_{pref}\subseteq A_{i+1}$ such that $(A_{i+1},pref)$ is chain recurrent. 
Taking $C\subseteq A_{i+1}$ which appears infinitely often among the $C_{pref}$ concludes the point.

For all $i$, $A_i\stackrel{\rho_i}{\rightarrow}A_i'$ is such that we can find $B_i\subseteq A_i$ such that 
$(B_i,\rho_i)$ is chain recurrent. Since for all $i$ we have that $A_i'\stackrel{\rho_i'}{\rightarrow}A_{i+1}$ is chain recurrent by construction, 
we get a sub-sequence of recurrent execution trees of $(\alpha,w)$ of length at most $2^{|Q|}$ (only the sub-sequences which correspond to arrows 
$\stackrel{a_i}{\rightarrow}$ may not contain a chain recurrent sub-sequence).
\end{proof}

\subsection{Details of Sub-Section \ref{ss-dec-pbs-ssPA}}

\smallskip\noindent{\bf Details of Proposition \ref{p-simple-equiv-reach}.} 
We prove Proposition \ref{p-simple-equiv-reach}.

\begin{proof}
We first show that \textbf{(1)}$\Rightarrow$\textbf{(2)}. 

 Let $\mathcal{A}$ be a PA, let $\rho\in\Sigma^*$, let $C\subseteq Q$, and let $\mathcal{H}_\mathcal{A}$ be the extended support graph of $\mathcal{A}$. 
Suppose $C\stackrel{\#-\rho}{\rightarrow}D$.
Let $B=(b_1,...,b_n)$ be a chain of borders and $i\in1,...,|\rho|$ be such that given $\LG(B,\LG(C,\rho,\mathcal{A}))=(G_1,...,G_n)$ we have $D=\dest(G_1,...,G_i)$. Given $j\in1,...,n$, let 
$\LG(b_1,...,b_j,\LG(C,\rho,\mathcal{A}))=(G_1^j,...,G_n^j)$. By induction we can show that for all $j\in1,...,n$ and all $i\in1,...,|\rho|$ we have that $\dest(G_1^j,...,G_i^j)$ 
is limit reachable from $C$.

We now show that \textbf{(2)}$\Rightarrow$\textbf{(1)}. 
Suppose that $F$ is limit reachable from $\alpha$ in $\mathcal{A}$. First, if there exists $\rho\in\Sigma^*$ such that $\mathrm{Supp}(\delta(\alpha,\rho))\subseteq F$, 
then by definition $F$ is reachable from $\mathrm{Supp}(\alpha)$ in $\mathcal{H}_\mathcal{A}$. 

Suppose now that $F$ is limit reachable from $\alpha$ in $\mathcal{A}$, but that for all $\rho\in\Sigma^*$ we have $\mathrm{Supp}(\delta(\alpha,\rho))\not\subseteq F$. 
We define the following probabilistic automaton $\mathcal{B}$ with state space $Q'$, alphabet $\Sigma'$ and transition function as follows:
\begin{compactitem}
 \item $Q'=Q\cup\lbrace \perp\rbrace$ where $\perp$ is a new state.
 \item $\Sigma'=\Sigma\cup\lbrace e\rbrace$ where $e$ is a new symbol.
 \item We keep the same transitions on $\mathcal{B}$ as in $\mathcal{A}$ when the labels are in $\Sigma$. 
 \item Given $q\in F$, we add an extra transition with label $e$ which leads to a state $q\in Q$ with probability $\alpha(q)$.
 \item Given $q\in Q\setminus F$, we add an extra transition with label $e$ which leads to state $\perp$ with probability one.
 \item From state $\perp$, given any $a\in\Sigma'$ we loop with probability one on $\perp$.
\end{compactitem}
We show that the automaton $\mathcal{B}$ is not simple. Since $F$ is limit reachable from $\alpha$, we can let $\lbrace\rho_n\rbrace_{n\in\mathbb{N}}$ 
be a sequence of finite words such that for all $n\in\mathbb{N}$ we have $\delta(\alpha,\rho_n)(F)>1-\dfrac{1}{2^n}$. We define $w\in\Sigma^\omega$ as:
\[w=\rho_1\cdot e\cdot\rho_2\cdot e\cdot \rho_3...\]
We claim that the process induced on the state space $Q'$ of $\mathcal{B}$ by $w$ is not simple. 
First, notice that at any time, if the current distribution of the process is $\beta\in\Delta(Q')$ and the letter $e$ is taken as input, 
then the probability to be in a state $q\in Q$ at the next step is equal to $\alpha(q)*\beta(F)$. As a consequence, the automaton $\mathcal{A}$ is not structurally simple.

Given $k\in\mathbb{N}$, let $\beta_k\in\Delta(Q')$ be the distribution on $Q'$ that we get after having read $\rho_1\cdot e\cdot\rho_2\cdot e\ldots \rho_k\cdot e$. 
By the choice of the $\rho_n$, for all $k$ we have $\beta_k(Q)>1/2$. 
By hypothesis, there exists $q,q'\in Q$ such that $q\in\mathrm{Supp}(\alpha)$ and such that for an infinite number of 
$n\in\mathbb{N}$ we have $q'\in\mathrm{Supp}(\delta(q,\rho_n))$ and $q'\not\in F$. Let $\gamma=\alpha(q)$. We have found a couple $q,q'\in Q$ such that:
\begin{compactitem}
 \item For all $k$ we have $\beta_k(q)>\gamma/2$
 \item Infinitely often, $\delta(q,\rho_{k})(q')>0$ and $\delta(\alpha,\rho_1\cdot e\ldots \cdot e\cdot\rho_{k})(q')<\dfrac{1}{2^k}$
\end{compactitem}

Such a couple $q,q'$ invalidates the Proposition \ref{p-separation simple jets} which holds for simple process. 
Indeed, by Proposition \ref{p-separation simple jets}, if infinitely often we have $\mu^w_n(q)>\gamma$, then there exists $N\in\mathbb{N}$ and $\gamma'>0$ 
such that for all $n_2>n_1\geq N$, if $\mu^w_{n_1}(q)>\gamma$ and $\delta(q,w^{n_2}_{n_1+1})(q')>0$, then $\mu^w_{n_2}(q')>\gamma'$. 
Thus, $\mathcal{B}$ is not simple. 

By definition, this implies that there exist $C\subseteq Q$, $D\subseteq D$ and $\rho\in\Sigma^*$ such that $C\stackrel{\#-\rho}{\rightarrow}D$ and $D$ is minimal, and $(D,\rho)$ is not chain recurrent. 
Since the automaton $\mathcal{A}$ is supposed to be structurally simple, $\rho$ must contain the letter $e$. Since $(D,\rho)$ is not chain recurrent, $D\not=\lbrace\perp\rbrace$, and in fact, since $D$ is minimal, we 
must have that $\perp\not\in D$. Let $B=(b_1,...,b_n)$ be the chain of borders such that $\dest(\LG(B,\LG(\rho,C,\mathcal{A})))=D$, let $\mathcal{I}_0=\LG(\rho,C,\mathcal{A})$, and for all 
$i\in1,...,n$ let $\mathcal{I}_i=\LG(b_i(\LG(b_{i-1}(...\LG(b_1(\LG(C,\rho,\mathcal{A})))...))))$. Let $k=|\rho|$, and for all $i\in1,...,n$ we write $\mathcal{I}_i=(G_0,...,G_k)$. 
Let $\rho=a_1,...,a_k$, and let $j_0$ be the largest integer in $1,...,k$ such that $a_{j_0}=e$. Then we have $A_{j_0}\subseteq F$. As a consequence, $F$ is reachable from $\Supp(\alpha)$ in $\mathcal{H}_\mathcal{A}$.

We get that in he new automaton, $\Supp(\alpha)\stackrel{\#}{\rightarrow}\Supp(\alpha)$, and more precisely that there exists a word $\rho\in\Sigma^*$ which contains the 
letter $e$ and which is such that $\Supp(\alpha)\stackrel{\#-\rho}{\rightarrow}\Supp(\alpha)$. This implies that $F$ is reachable from $\Supp(\alpha)$ in $\mathcal{H}_\mathcal{A}$.

\end{proof}

\smallskip\noindent{\bf Details of Theorem \ref{prop:simple-dec}.} 
We prove Theorem \ref{prop:simple-dec}.
\begin{proof}
 We can prove the following, using the same kind of arguments as in the previous proofs: 
a structurally simple PA $\mathcal{A}$ satisfies the qualitative limit parity problem iff there exists a set of states $A\subseteq Q$ such that:
\begin{compactitem}
 \item $A$ is limit reachable from $\mathrm{Supp}(\alpha)$
 \item There exists $\rho\in\Sigma^*$ of length at most $2^{|Q|}$ such that $A\cdot\rho\subseteq A$ and the parity condition is satisfied on the Markov chain induced by $A,\rho$.
\end{compactitem}
This condition can be checked in EXPSPACE, using Theorem \ref{t-limit-reach-PSPACE}.
\end{proof}

\smallskip\noindent{\bf Details of Theorem \ref{t-decid-ssimple}.} 
We prove Theorem \ref{t-decid-ssimple}.

We first need a preliminary Lemma:

\begin{lemma}\label{l-sous-reach}
Les $\mathcal{A}$ be a probabilistic automaton. Let $C,D\subseteq Q$, and suppose that $C\stackrel{\#}{\rightarrow}D$. Let $E\subseteq C$. Then there exists $F\subseteq D$ such that $E\stackrel{\#}{\rightarrow}F$.
\end{lemma}
\begin{proof}
 Let $\rho\in\Sigma^*$ and $B$ be a chain of reduction such that given $\LG(B,\LG(C,\rho,\mathcal{A}))=(G_1,...,G_n)$ there exists $i\in1,...,n$ such that $\dest(G_1,...,G_i)=D$. 
Then by  applying the same chain of border on $\LG(E,\rho,\mathcal{A})$ we get, if $\LG(B,\LG(E,\rho,\mathcal{A}))=(G_1',...,G_n')$, that $\dest(G_1',...,G_i')\subseteq D$. 
\end{proof}

We now give the proof of Theorem \ref{t-decid-ssimple}.
\begin{proof}
Let $\mathcal{A}$ be a probabilistic automaton. We say that a subset $C$ of $Q$ is \emph{minimal} if there exists no $D\subsetneq C$ such that $C\stackrel{\#}{\rightarrow}D$. Given $C\subseteq Q$ minimal, let
\[S(C)=\lbrace A\subseteq Q\ |\ \exists\rho\in\Sigma^*\ s.t.\ A\stackrel{\#-\rho}{\rightarrow}C\ and\ \Supp(\delta(A,\rho))\not=C\]
We claim that the following properties are equivalent:
\begin{enumerate}
 \item $\mathcal{A}$ is simple
 \item Given $C\subseteq Q$ which is minimal, for all $\rho\in\Sigma^*$, possibly empty, we have $\Supp(\delta(C,\rho))\not\in S(C)$
\end{enumerate}
This would prove the result since given $C$ we can decide in EXPSPACE whether $C$ is minimal, and we can compute $S(C)$ in EXPTIME using Algorithm \ref{a-comp_SG}.

We prove the claim. Suppose first $(1)$, and let $C\subseteq Q$ be minimal, and suppose that there exists $\rho\in\Sigma^*$ such that $D=\Supp(\delta(C,\rho))\in S(C)$. Then there would exist $\rho'\in\Sigma^*$ such 
that $D\stackrel{\#-\rho'}{\rightarrow}C\ and\ \Supp(\delta(D,\rho))\not=C$. We get that 
\[C\stackrel{\rho}{\rightarrow} D\stackrel{\#-\rho'}{\rightarrow}C\ and\ \Supp(\delta(D,\rho))\not=C\]
Since $C$ is minimal, by definition of a structurally simple automaton, we get that $C,\rho\cdot\rho'$ is chain recurrent. As a consequence, we should have $D,\rho'$ chain recurrent, 
which is a contradiction. Hence $(1)$ implies $(2)$.

Conversely, suppose $(2)$, and let $\rho\in\Sigma^*$, $C\subseteq Q$, and $D\subseteq C$ be minimal among the $D\subseteq Q$ such that $C\stackrel{\#-\rho}{\rightarrow}D$. Suppose that $D,\rho$ is 
not chain recurrent. By Lemma \ref{l-sous-reach}, and since $D$ is minimal, 
there exists $\rho\in\Sigma^*$ such that $D\stackrel{\#-\rho}{\rightarrow}D$. Thus, if $D,\rho$ is not chain recurrent we get $D\in S(D)$. This contradicts $(2)$. Hence we get 
that $(2)$ implies $(1)$.
\end{proof}

\subsection{Details of Sub-Section \ref{ss-clos-props}}

\smallskip\noindent{\bf Details of the product construction and Proposition \ref{p-product simple}.} 

Given $\mathcal{A}_1=(S_1,\Sigma,\delta_1,\alpha_1)$ and $\mathcal{A}_2=(S_2,\Sigma,\delta_2,\alpha_2)$ two structurally simple automata on the same alphabet $\Sigma$, the construction of the product automaton 
$\mathcal{A}_1\Join\mathcal{A}_2=(S,\Sigma,\delta,\alpha)$ is as follows:
\begin{compactitem}
 \item $S$ is the Cartesian product of $S_1$ and $S_2$: $S=S_1\times S_2$.
 \item Given $(s_1,s_2),(s_1',s_2')\in S$ and $a\in\Sigma$, $\delta((s_1,s_2),a)((s_1',s_2'))=\delta_1(s_1,a)(s_1')\cdot\delta_2(s_2,a)(s_2')$.
 \item Given $(s_1,s_2)\in S$, $\alpha((s_1,s_2))=\alpha_1(s_1)\cdot\alpha_2(s_2)$.
\end{compactitem}

Given $s=(s_1,s_2)\in S$, let $p_1(s)=s_1$ and $p_2(s)=s_2$ be the respective projections of $s$ on the state spaces of $\mathcal{A}_1$ and $\mathcal{A}_2$.

We now prove Proposition \ref{p-product simple}.

\begin{proof}

For this we just have to remark that given $\rho\in\Sigma^*$, given $A\subseteq S$, if we let $A_1=\lbrace s\in S_1\ |\ \exists t\in S_2\ s.t.\ (s,t)\in S\rbrace$ and 
$A_2=\lbrace t\in S_2\ |\ \exists s\in S_s\ s.t.\ (s,t)\in S\rbrace$, then we get $A\cdot\rho=(A_1\cdot\rho)\times(A_2\cdot\rho)$. Moreover, if $A\cdot\rho=A$, then 
we have that $A\cdot\rho^\#=(A_1\cdot\rho^\#)\times(A_2\cdot\rho^\#)$. This implies that the structural simplicity condition is satisfied on $\mathcal{A}_1\Join\mathcal{A}_2$ iff 
it is satisfied both on $\mathcal{A}_1$ and $\mathcal{A}_2$.

\end{proof}

\smallskip\noindent{\bf Details of Theorem \ref{t-class-robust}.} 
We prove Theorem \ref{t-class-robust}.

\begin{proof}
First, remark that the stability of the class of languages recognized by structurally simple parity automata under the positive semantics is trivial: 
we just consider a ``union automaton'' whose structure is the union of the structures of the two given automata, and whose initial distribution is a mix of the two given automata initial distributions.

We consider now the stability of this class of language under the intersection operator. Let $\mathcal{A}$ be a structurally simple parity automaton. 
By defining a relevant set of accepting sets, we can transform its accepting condition to transform it to a positive Street PA which recognizes the same language. 
Since we do not change the structure of the automaton nor its transition function, the new automaton is still structurally simple. 
Now, given two Street PA with the positive semantics, using a classical product construction, we can construct a Street PA which, under the positive semantics, 
accept a language which is the intersection of the languages of the two Street automata. 
By Proposition \ref{p-product simple}, this Street PA is still structurally simple. 
Finally, using a construction a la Safra, we can construct a parity PA which, under the positive semantics, recognizes the same language as the last Street PA. 
We can show that the construction a la Safra keeps the automaton structurally simple, since it can be seen as a product construction which does 
not add any probabilistic transition. We get the stability of the languages of positive parity PA under union and intersection.

Remark next that PAs with positive parity semantics and PAs with almost parity semantics are dual of each others: 
given a PA $\mathcal{A}$ with positive parity semantics, by inverting the parity condition (taking a new parity function $p'=p-1$), 
we get a new PA $\mathcal{A}'$ whose language is the complementary of $\mathcal{L}^{>0}(\mathcal{A})$:  $\mathcal{L}^{=1}(\mathcal{A}')=\mathcal{L}^{>0}(\mathcal{A})^c$. 
As a consequence, if the class of languages recognized by positive parity PAs is stable under intersection and complementation, so is the class of languages recognized by almost parity PAs.
\end{proof}

\subsection{Details of Sub-Section \ref{s-subclasses}}

\smallskip\noindent{\bf Details of Proposition \ref{p-extends_acycl}.} 
We prove Proposition \ref{p-extends_acycl}.

\begin{proof}

By contraposition, let $\mathcal{A}$ be an automaton which is not structurally simple. Let $C\subseteq Q$, $\rho\in\Sigma^*$, and $D\subseteq Q$ minimal such that 
$C\stackrel{\#-\rho}{\rightarrow}D$, and such that $D,\rho$ is not chain recurrent. Let $E=\cup_{i\geq0}\Supp(\delta(D,\rho^i))$. We have $E\stackrel{\rho}{\rightarrow}E$.
Let $\rho=\rho_1\cdot\rho_2\cdot\rho_3$ be a decomposition such that $(\Supp(\delta(D,\rho_1)),\rho_2)$ is a $\#$-reduction. Then we have that 
$(\Supp(\delta(E,\rho_1),\rho_2)$ is also a $\#$-reduction. As a consequence, we can build a cycle in the support graph of $\mathcal{A}$, and thus $\mathcal{A}$ is not $\#$-acyclic.

Notice that every deterministic automaton is structurally simple. This implies that some structurally simple automata (for instance deterministic automata whose 
support graph contains a cycle) are not $\#$-acyclic.


\end{proof}

\smallskip\noindent{\bf Details of Proposition \ref{p-extends_hierch}.} 
We prove Proposition \ref{p-extends_hierch}.

\begin{proof}
Let $\mathcal{A}$ be a $k$-hierarchical automata. Let $C=B_1\rightarrow B_2\rightarrow...\rightarrow B_l=B_1$ be an elementary cycle in the support graph of $\mathcal{A}$, i.e. 
a cycle which does not contain any sub-cycle. Let $i,j\in\set{1,\ldots,l}$, let $q\in B_i$, and let $q'\in B_j$. Then 
there exists $\rho\in\Sigma^*$ such that $\delta(q,\rho)(q')>0$, and there exists $\rho'\in\Sigma^*$ such that $\delta(q,',\rho')(q)>0$. 
This implies that $\mathrm{rk}(q)=\mathrm{rk}(q')$. This implies that there is not probabilistic transition in the cycle. 
As a consequence, $C$ can not contain a $\#$-reduction. This proves that $\mathcal{A}$ is structurally simple. The automaton of Example \ref{e-simple-autom} is structurally simple, 
but not hierarchical, which completes the proof.
\end{proof}

\end{document}